\documentclass[%
 aip,
 jmp,%
 amsmath,amssymb,
]{revtex4-1}

\usepackage{amsmath,amssymb}
\usepackage{enumerate}
\usepackage{braket}
\usepackage{bbm}
\usepackage{dsfont}

\usepackage{amsthm}
\newtheorem{theorem}{Theorem}
\newtheorem{lemma}{Lemma}
\newtheorem{proposition}{Proposition}
\newtheorem{corollary}{Corollary}

\theoremstyle{definition}
\newtheorem{definition}{Definition}
\newtheorem{remark}{Remark}

\DeclareMathOperator*{\swlim}{\sigma w-lim}

\newcommand{\defarrow}{\stackrel{\mathrm{def.}}{\Leftrightarrow}}

\newcommand{\cmplx}{\mathbb{C}}
\newcommand{\realn}{\mathbb{R}}

\newcommand{\natn}{\mathbb{N}}
\newcommand{\natz}{\mathbb{N}_0}

\newcommand{\Real}{\mathrm{Re}} 
\newcommand{\Imag}{\mathrm{Im}} 

\newcommand{\calL}{\mathcal{L}}
\newcommand{\cH}{\mathcal{H}}
\newcommand{\Hin}{\mathcal{H}_\mathrm{in}  }

\newcommand{\cK}{\mathcal{K}}
\newcommand{\LH}{\mathcal{L} (\mathcal{H}) }
\newcommand{\LHin}{\mathcal{L} (\mathcal{H}_\mathrm{in}  ) }

\newcommand{\LK}{\mathcal{L}  (\mathcal{K})}

\newcommand{\LHF}{\mathcal{L}  ( \mathcal{H}_F )}

\newcommand{\cstar}{$C^\ast$}
\newcommand{\A}{\mathcal{A}}
\newcommand{\B}{\mathcal{B}}
\newcommand{\C}{\mathcal{C}}
\newcommand{\D}{\mathcal{D}}
\newcommand{\Ain}{\mathcal{A}_{\mathrm{in}}}

\newcommand{\Ainf}{\A^{\otimes \infty}}
\newcommand{\Atn}{\A^{\otimes n}}

\newcommand{\M}{\mathcal{M}}
\newcommand{\N}{\mathcal{N}}
\newcommand{\Min}{\mathcal{M}_{\mathrm{in}}}

\newcommand{\tL}{\widetilde{\Lambda}}
\newcommand{\tG}{\widetilde{\Gamma}}

\newcommand{\tpi}{\widetilde{\pi}}
\newcommand{\trho}{\widetilde{\rho}}
\newcommand{\tphi}{\widetilde{\varphi}}
\newcommand{\tomega}{\widetilde{\omega}}

\newcommand{\tPsi}{\widetilde{\Psi}}

\newcommand{\tA}{\widetilde{\mathcal{A}}}

\newcommand{\bL}{\overline{\Lambda}}

\newcommand{\bpi}{\overline{\pi}}
\newcommand{\bTh}{\overline{\Theta}}

\newcommand{\vph}{\varphi}

\newcommand{\cS}{\mathcal{S}}

\newcommand{\Sp}{\mathcal{S}_{\mathrm{p}}} 
\newcommand{\Ssep}{\mathcal{S}_{\mathrm{sep}}} 
\newcommand{\SA}{\mathcal{S}(\mathcal{A})}

\newcommand{\chset}[2]{\mathbf{Ch} (#1 \to #2)}

\newcommand{\nchset}[2]{\mathbf{Ch}_\sigma (#1 \to #2)}

\newcommand{\chseteb}[2]{\mathbf{Ch}^{\mathrm{EB}} (#1 \to #2)}
{}

\newcommand{\oM}{\mathsf{M}}

\newcommand{\atensor}{\otimes_{\mathrm{alg}}}

\newcommand{\id}{\mathrm{id}}
\newcommand{\unit}{\mathds{1}}

\newcommand{\cocp}{\preccurlyeq_{\mathrm{CP}}}
\newcommand{\eqcp}{\sim_{\mathrm{CP}}}

\newcommand{\CH}{\mathcal{C}\mathcal{H}}
\newcommand{\CHqc}{\mathcal{C}\mathcal{H}^{\mathrm{QC}}}
\newcommand{\CHeb}{\mathcal{C}\mathcal{H}^{\mathrm{EB}}}

\newcommand{\FI}{\mathbb{F}(I)}

\newcommand{\cco}{\overline{\mathrm{co}}}

\newcommand{\bfMpo}{\mathbf{M}_{+,1}}

\newcommand{\Sn}{\mathfrak{S}_n}

\newcommand{\LF}{\Lambda_F}

\newcommand{\gbarg}{\Gamma^{\mathrm{B}}}

\begin{document}
\preprint{}
\title[]{%
Entanglement-breaking channels with general outcome operator algebras%
}
\author{Yui Kuramochi}
\affiliation{Laboratory of Quantum Engineering and Quantum Metrology, School of Physics and Astronomy, Sun Yat-Sen University (Zhuhai Campus), Zhuhai 519082, China}
\email{yui.tasuke.kuramochi@gmail.com}
\date{}

\begin{abstract}
A unit-preserving and completely positive linear map, or a channel,
$\Lambda \colon \mathcal{A} \to \mathcal{A}_{\mathrm{in}}$
between $C^\ast$-algebras $\mathcal{A}$ and $\mathcal{A}_{\mathrm{in}}$ 
is called entanglement-breaking (EB) if 
$\omega \circ( \Lambda \otimes \mathrm{id}_{\mathcal{B}} ) $
is a separable state for any $C^\ast$-algebra $\mathcal{B}$   
and any state $\omega$ on 
the injective $C^\ast$-tensor product $\mathcal{A}_{\mathrm{in}} \otimes \mathcal{B} .$
In this paper, we establish the equivalence of the following conditions for a channel 
$\Lambda$ with a quantum input space and 
with a general outcome $C^\ast$-algebra, generalizing 
known results in finite dimensions:
(i) $\Lambda$ is EB; 
(ii) $\Lambda$ has a measurement-prepare form (Holevo form);
(iii) $n$ copies of $\Lambda$ are compatible for all $2 \leq n < \infty ;$
(iv) countably infinite copies of $\Lambda$ are compatible.
By using this equivalence, 
we also show that the set of randomization-equivalence classes of 
normal EB channels with a fixed input von Neumann algebra 
is upper and lower Dedekind-closed, 
i.e.\ the supremum or infimum of 
any randomization-increasing or decreasing net of EB channels 
is also EB.
As an example, we construct an injective normal EB channel with an arbitrary outcome operator algebra $\mathcal{M}$ acting on an infinite-dimensional separable Hilbert space
by using the coherent states and the Bargmann measure.
\end{abstract}
\pacs{03.67.-a, 02.30.Tb}
\keywords{entanglement-breaking channel, Holevo form, 
channel compatibility, randomization order, no-broadcasting theorem, Dedekind-closedness}
\maketitle

\section{Introduction} \label{sec:intro}
A quantum channel $\Lambda$ is called entanglement-breaking (EB)
if $\Lambda$ tensored with the identity on any state space maps
any entangled state to a separable state.
In Ref.~\onlinecite{Horodecki2003} some useful characterizations
of finite-dimensional EB channels were established.
Later EB condition is considered in the infinite-dimensional quantum 
systems in Refs.~\onlinecite{holevo2005separability,0036-0279-60-2-L12,Holevo2008,Holevo2011,He2013}
and the operator algebraic setup in 
Refs.~\onlinecite{STORMER20082303,Stoermer_2009,doi:10.1063/1.5024385}.
Among them, in Ref.~\onlinecite{holevo2005separability},
it was shown that a channel with input and outcome separable Hilbert spaces
is EB if and only if the channel has a measurement-prepare form (Holevo form).
Recently,  
another characterization of the EB condition in terms of 
channel compatibility was given in Ref.~\onlinecite{1751-8121-50-13-135302}
(Proposition~12):
a finite-dimensional channel is EB if and only if 
$n$ copies of the channel have a joint channel
for all $2 \leq n < \infty .$
The proof in Ref.~\onlinecite{1751-8121-50-13-135302} is based on 
an upper bound~\cite{10.1007/978-3-642-18073-6_2} 
on the diamond norm distance 
that explicitly depends on the dimension of the system
and is not applicable to the infinite-dimensional case.

The purpose of this paper is to generalize these characterizations of the EB condition
to the case of channels with (possibly non-separable) quantum input spaces
and with general outcome operator algebras.
The main finding of this paper is Theorem~\ref{theo:main},
which, roughly speaking, establishes the equivalence of 
the following conditions for a channel $\Lambda$
with quantum input and operator algebraic outcome spaces:
\begin{enumerate}[(i)]
\item
$\Lambda$ is EB;
\item
$\Lambda$ has a measurement-prepare form;
\item
$n$-copies of $\Lambda$ are compatible for all $2 \leq n < \infty ;$
\item
countably infinite copies of $\Lambda$ are compatible.
\end{enumerate}
The equivalence of the first three conditions generalizes known results in finite-dimensions,~\cite{Horodecki2003,1751-8121-50-13-135302}
while the last infinite-self-compatibility condition is a new characterization.

This paper is organized as follows.
After mathematical preliminaries in Section~\ref{sec:prel},
we consider in Section~\ref{sec:sep} general properties of separable states on
injective \cstar-tensor product algebras and show 
that any separable state is represented as a barycentric integral of 
product states (Corollary~\ref{coro:sep}),
which is a \cstar-algebra version of the corresponding integral representation 
given in Ref.~\onlinecite{holevo2005separability}.
In Section~\ref{sec:main} we prove Theorem~\ref{theo:main}, 
the main result of this paper.
In Section~\ref{sec:dedekind}, 
as an application of the characterizations given in Theorem~\ref{theo:main},
we show that
the supremum or infimum of 
any randomization-increasing or decreasing net of normal EB channels with 
a fixed input von Neumann algebra is also EB
(Theorem~\ref{theo:dc}).
In Section~\ref{sec:ex} we construct a normal injective EB channel
that has an arbitrary outcome von Neumann algebra acting on a separable 
Hilbert space.
Section~\ref{sec:conclusion} concludes the paper.

\section{Preliminaries} \label{sec:prel}
This section introduces mathematical preliminaries needed in the main part
and fixes the notation.
For general references of operator algebra we refer to
Refs.~\onlinecite{takesakivol1,sakaibook}.

\subsection{States and channels on operator algebras}
Throughout this paper every \cstar-algebra $\A$ is assumed to have 
a unit element which we write as $\unit_\A .$
For a \cstar-algebra $\A ,$ a positive and normalized linear functional
$\phi \in \A^\ast$ is called a state on $\A .$
We denote by $\cS (\A) $ the set of states on $\A ,$
which is a compact convex subset of the dual space $\A^\ast$
in the weak-$\ast$ topology $\sigma (\A^\ast , \A)$.
We also denote by $\Sp (\A)$ the set of pure states on $\A .$
Since $\Sp (\A)$ coincides with the set of extreme points of $\cS (\A) ,$
the Krein-Milman theorem implies that $\cS (\A) = \cco ( \Sp (\A) ) ,$
where $\cco (\cdot)$ denotes the closed convex hull in the weak-$\ast$ topology.
For a Hilbert space $\cH ,$
$\LH$ and $\unit_\cH$ denote the algebra of bounded linear operators 
and the unit operator on $\cH ,$
respectively.
We also write the inner product of a Hilber space $\cH$ as $\braket{\xi | \eta } $
$(\xi , \eta \in \cH) ,$
which is linear and anti-linear with respect to $\eta$ and $\xi ,$ respectively.

A channel (in the Heisenberg picture)
is a unit-preserving and completely positive (CP) map
$\Lambda \colon \A \to \B$
between \cstar-algebras $\A$ and $\B .$
The algebras $\A$ and $\B$ are called the outcome and input spaces, or algebras, 
of $\Lambda ,$
respectively.
In the Schr\"odinger picture, an input state
$\phi \in \cS (\B)$ is mapped to the outcome state
$\Lambda^\ast (\phi) = \phi \circ \Lambda \in \cS (\A) .$
For \cstar-algebras $\A$ and $\B ,$
the set of channels from $\A $ to $\B$ is denoted by $\chset{\A}{\B} .$

A channel $\Lambda \in \chset{\M}{\N}$ between von Neumann algebras
$\M$ and $\N$ is called normal if $\Lambda$ is continuous with respect to
the $\sigma$-weak topologies of $\M$ and $\N ,$ respectively.
The set of normal channels from $\M$ to $\N$ is denoted by
$\nchset{\M}{\N}  .$
A normal channel with a commutative outcome von Neumann algebra is called
a quantum-classical (QC) channel.

Let $\Lambda \in \chset{\A}{\Ain}$
and $\Gamma \in \chset{\B}{\Ain}$
be channels with the common input space $\Ain .$
We define the randomization (or coarse-graining or concatenation)
relations~\cite{1751-8121-50-13-135302,kuramochi2018incomp}
for channels as follows:
\begin{itemize}
\item
$\Lambda \cocp \Gamma$ 
($\Lambda$ is a randomization of $\Gamma$)
$: \defarrow$
there exists a channel $\alpha \in \chset{\A}{\B}$
such that $\Lambda = \Gamma \circ \alpha ;$
\item
$\Lambda \eqcp \Gamma $
($\Lambda$ is randomization-equivalent to $\Gamma$)
$:\defarrow$
$\Lambda \cocp \Gamma$ and $\Gamma \cocp \Lambda .$ 
\end{itemize}
The relations $\cocp $ and $\eqcp$ are binary preorder and equivalence relations,
respectively,
defined on the class of channels with a fixed input algebra.
In the above definition, if $\A , $ $\B ,$ and $\Ain$ are von Neumann algebras
and $\Lambda $ and $\Gamma$ are normal, 
$\Lambda \cocp \Gamma$
if and only if
there exists a normal channel
$\alpha \in \nchset{\A}{\B}$ such that
$\Lambda = \Gamma \circ \alpha .$~\cite{kuramochi2018incomp} 

Let $\A$ be a \cstar-algebra
and let $\pi_\A \colon \A \to \calL (\cH_\A)$ be the universal representation of $\A ,$
which is the direct sum of the Gelfand-Naimark-Segal
(GNS) representations of 
all the states on $\A .$
The von Neumann algebra $\pi_\A (\A)^{\prime \prime} ,$
where the prime denotes the commutant, 
generated by $\pi_\A (\A)$ is called the universal enveloping von Neumann algebra 
of $\A .$
The algebra $\pi_\A (\A)^{\prime \prime}$ is isometrically isomorphic to 
the double dual space $\A^{\ast \ast}$
and we identify $\pi_\A (\A)^{\prime \prime}$ with $\A^{\ast \ast } .$
By identifying $\A$ with $\pi_\A (\A) ,$
we also regard $\A$ as a $\sigma$-weakly dense 
\cstar-subalgebra of $\A^{\ast \ast} .$

Let $\A$ be a \cstar-algebra and let $\Min$ be a von Neumann algebra.
Then any channel $\Lambda \in \chset{\A}{\Min}$
uniquely extends to a normal channel
$\bL \in \nchset{\A^{\ast \ast }}{\Min} $
which we call the normal extension of $\Lambda .$\cite{kuramochi2018incomp}
The normal extension $\bL$ is the least normal channel that upper bounds $\Lambda$
in the randomization order $\cocp .$
Specifically, $\Lambda \cocp \bL$ and for any normal channel 
$\Gamma \in \nchset{\N}{\Min} ,$
$\Lambda \cocp \Gamma$ if and only if
$\bL \cocp \Gamma $
(Ref.~\onlinecite{kuramochi2018incomp}, Lemma 7).
This also implies that $\Lambda \eqcp \bL$ if $\Lambda$ is normal.

\subsection{Positive-operator valued measures}
A positive-operator valued measure (POVM)~\cite{davieslewisBF01647093,holevo2011probabilistic}
on a Hilbert space $\cH$ is a triple $(\Omega ,\Sigma , \oM)$
such that $(\Omega , \Sigma)$ is a measurable space, i.e.\
$\Sigma$ is a $\sigma$-algebra on a set $\Omega ,$
and $\oM \colon \Sigma \to \LH$ is a map satisfying
\begin{enumerate}[(i)]
\item
$\oM (E) \geq 0$
$(E \in \Sigma) ;$
\item
$\oM (\Omega) = \unit_{\cH} ;$
\item
for any disjoint sequence $\{ E_n \}_{n \in \natn} \subset \Sigma ,$
$\oM (\bigcup_{n \in \natn} E_n) = \sum_{n \in \natn} \oM(E_n) .$
\end{enumerate}
A POVM describes the classical outcome statistics of a general quantum measurement.
For a measurable space $(\Omega , \Sigma) ,$
we denote by $B(\Omega , \Sigma)$ the set of bounded, complex valued, 
$\Sigma$-measurable functions on $\Omega ,$
which is a commutative \cstar-algebra under the supremum norm
$\| f \| := \sup_{x \in \Omega} | f(x) | .$
If $(\Omega , \Sigma , \oM)$ is a POVM on $\cH ,$
for each $f \in B(\Omega , \Sigma) $ we can define the integral
$\int_{\Omega} f (x) d \oM (x) \in \LH$
in the weak sense.

Let $\Omega$ be a compact Hausdorff space.
We denote by $C (\Omega)$ and $\B (\Omega)$
the commutative \cstar-algebra of bounded complex-valued continuous functions on
$\Omega $
and the Borel $\sigma$-algebra generated by the set of open subsets of $\Omega ,$
respectively.
A POVM $(\Omega , \B (\Omega) , \oM)$ on a Hilbert space $\cH$
is called regular
if 
\[
	\B (\Omega) \ni E 
	\longmapsto 
	\braket{\xi | \oM(E) \eta}
	\in \cmplx
\]
is a regular complex measure for any $\xi , \eta \in \cH .$
According to the Riesz-Markov-Kakutani representation theorem
we have a one-to-one correspondence 
\[
	\cS (C(\Omega)) \ni \phi 
	\longmapsto 
	\mu_\phi \in 
	\bfMpo (\Omega)
\]
between the state space $\cS (C(\Omega))$
and the set $\bfMpo (\Omega)$ of regular probability measures on $\Omega$
such that
\[
	\phi (f)
	=
	\int_\Omega 
	f(x) d\mu_\phi (x) 
	\quad
	(f \in C(\Omega)) .
\]
A similar representation theorem also holds for 
channels as follows.
\begin{proposition}[Ref.~\onlinecite{busch2016quantum}, Theorem~4.4]
\label{prop:RMK}
Let $\Omega$ be a compact Hausdorff space and let $\cH$ be a Hilbert space.
Then for any channel $\Gamma \in \chset{C(\Omega)}{\LH}$
there exists a unique regular POVM $(\Omega , \B (\Omega) , \oM)$ on $\cH$
such that
\[
	\Gamma (f)
	=
	\int_\Omega f (x) d\oM (x)
	\quad
	(f \in C (\Omega)) .
\]
\end{proposition}

\subsection{Tensor products}
Let $\A$ and $\B$ be \cstar-algebras.
We denote the algebraic tensor and the injective \cstar-tensor 
products by $\A \atensor \B$ and $\A \otimes \B ,$
respectively.
If $\A$ and $\B$ are faithfully represented on Hilbert spaces 
$\cH$ and $\cK ,$ respectively,
then $\A \otimes \B$ coincides with the norm closure
of $\A \atensor \B$ on $\calL (\cH \otimes \cK) $
up to isomorphism.
For \cstar-algebras $\A , $ $\B ,$ and $\C ,$
the associative law 
\[
(\A  \otimes \B) \otimes \C = \A \otimes(\B \otimes \C)
=:
\A \otimes \B \otimes \C
\]
holds
(again up to isomorphism)
and the \cstar-algebra $\A \otimes \B \otimes \C$
coincides with the norm closure of $\A \atensor \B \atensor \C$
on $\calL (\cH \otimes \cK \otimes \mathcal{J})$
if $\A ,$ $\B ,$ and $\C$ are faithfully represented on Hilbert spaces 
$\cH , $ $\cK ,$ and $\mathcal{J} ,$ respectively.  
The injective \cstar-tensor product of arbitrary finite number of 
\cstar-algebras can be defined similarly.

Let $\Lambda \in \chset{\A}{\B}$ and $\Gamma \in \chset{\C}{\D}$
be channels.
Then there exists a unique CP channel
$\Lambda \otimes \Gamma \in \chset{\A \otimes \C}{\B \otimes \D},$
called the tensor product channel,
such that
$\Lambda \otimes \Gamma ( A \otimes C) = \Lambda (A) \otimes \Gamma (C)$
$(A \in \A , C \in \C) .$

By the inductive limit procedure,
we can also define the injective \cstar-tensor product of 
infinite number of \cstar-algebras.~\cite{takeda1955}
In this paper we only need the countably infinite tensor product
$\A^{\otimes \infty}$ 
of identical \cstar-algebras, 
which is defined as follows.
For each \cstar-algebra $\A$
and each integer $1 \leq n < \infty ,$ we write
\[
	\A^{\otimes n}
	:=
	\underbrace{\A \otimes \dots \otimes \A}_{\text{$n$ elements}} .
\]
For each $1 \leq m <n < \infty ,$
the representation
$
\pi_{n \gets m} \colon
\A^{\otimes m} \ni X 
\mapsto
X \otimes \unit_{\A^{\otimes (n-m)}}
\in \A^{\otimes n}
$
is faithful and the family $(\pi_{n \gets m})_{1 \leq m < n < \infty}$
satisfies the consistency condition
$\pi_{n \gets k} = \pi_{n \gets m} \circ \pi_{m \gets k}$
$(1 \leq k <m < n < \infty) .$
Then there exist a \cstar-algebra
$\A^{\otimes \infty}$
and faithful representations $\pi_n \colon \A^{\otimes n} \to \A^{\otimes \infty}$
$(1 \leq n < \infty) $
such that
$\pi_m = \pi_n \circ \pi_{n\gets m}$
$(1 \leq m<n<\infty)$
and 
$\bigcup_{1 \leq n < \infty} \pi_n (\A^{\otimes n})$
is a norm dense $\ast$-subalgebra of $\A^{\otimes \infty} .$
Such $\Ainf$ and $(\pi_n)_{n \geq 1}$ are unique up to isomorphism.
The \cstar-algebra $\A^{\otimes \infty} $
is called the (countably) infinite injective \cstar-tensor product of $\A .$
The element $\pi_n (X) \in \A^{\otimes \infty}$ $(X \in \A^{\otimes n})$
and the \cstar-subalgebra $\pi_n (\A^{\otimes n}) \subset \A^{\otimes \infty}$
are written as $X \otimes \unit_{\A^{\otimes \infty}}$ and 
$\A^{\otimes n} \otimes \unit_{\A^{\otimes \infty}} ,$
respectively.

Let $\A$ and $\B$ be \cstar-algebras
and let $ 2 \leq n < \infty .$
A channel $\Lambda \in \chset{\A^{\otimes n}}{\B}$
is called symmetric if 
$\Lambda = \Lambda \circ \pi_\sigma$
for all $\sigma \in \Sn ,$
where $\Sn$ is the symmetric group of 
the finite set $\{ 1, \dots , n\} $
and $\pi_\sigma \colon \A^{\otimes n} \to \A^{\otimes n}$
$(\sigma \in \Sn)$
is a representation defined by
\[
	\pi_\sigma (A_1 \otimes \dots \otimes A_n)
	:=
	A_{\sigma(1)} \otimes \dots \otimes A_{\sigma(n)}
	\quad
	(A_1 , \dots , A_n \in \A) .
\]
The representation $\pi_\sigma$ 
$(\sigma \in \Sn )$
extends to the representation 
$\pi_\sigma \otimes \id_{\Ainf} \colon \Ainf \to \Ainf$ defined by
\[
	\pi_\sigma \otimes \id_{\Ainf}
	(X \otimes Y \otimes \unit_{\Ainf})
	:=
	\pi_\sigma (X) \otimes Y \otimes \unit_{\Ainf}
	\quad
	(1\leq m < \infty , X \in \A^{\otimes n } , Y \in \A^{\otimes m}).
\]
A channel $\tL \in \chset{\Ainf}{\B}$ is called symmetric 
if $\tL \circ (\pi_\sigma \otimes \id_{\Ainf}) = \tL$
for all $2 \leq n < \infty$
and all $\sigma \in \Sn .$
A channel $\tL \in \chset{\Ainf}{\B}$ is symmetric if and only if
the restriction
\[
	\tL_n \colon \Atn \ni X \longmapsto \tL (X \otimes \unit_{\Ainf}) \in \B
\]
is symmetric for all $2 \leq n <\infty .$

\subsection{Self-compatibility of channels}
Let $\Lambda \in \chset{\A}{\Ain}$ be a channel.
For $2 \leq n < \infty ,$
a channel $\Theta \in \chset{\Atn}{\Ain}$
is called an $n$-joint channel~\cite{1751-8121-50-13-135302}
of $\Lambda$ if 
\[
	\Lambda(A)
	=
	\Theta (A \otimes \unit_{\A^{\otimes (n-1)}})
	=
	\dots
	=
	\Theta (\unit_{\A^{\otimes k}} \otimes A \otimes \unit_{\A^{\otimes (n-k-1)}} )
	=\dots =
	\Theta ( \unit_{\A^{\otimes (n-1)}} \otimes A)
\]
for all $A \in \A ,$
i.e.\ the $n$ marginal channels of $\Theta$ are identical to $\Lambda .$
An $n$-joint channel $\Theta$ of $\Lambda ,$ if exists,
can always be taken to be symmetric by retaking $\Theta$
as~\cite{1751-8121-50-13-135302}
\[
	\frac{1}{n!}
	\sum_{\sigma \in \Sn}
	\Theta \circ \pi_\sigma .
\]
A channel $\Theta_\infty \in \chset{\Ainf}{\Ain}$ is called an 
$\omega$-joint channel of $\Lambda$ if
\[
	\Lambda (A)
	=
	\Theta_\infty 
	(A \otimes \unit_{\Ainf})
	=
	\Theta_\infty 
	( \unit_{\Atn} \otimes   A \otimes \unit_{\Ainf})
\]
for all $A \in \A$ and all $1 \leq n < \infty .$

\subsection{No-broadcasting theorem}
Let $\Lambda \in \chset{\A}{\Ain}$ be a channel 
and let $\A \otimes_x \A$ be either 
$\A \atensor \A $ or $\A \otimes \A .$
Then a channel $\Psi \in \chset{\A \otimes_x \A}{\A}$
is called a broadcasting channel of $\Lambda$ if
\[
	\Lambda (A) 
	=
	\Lambda \circ \Psi (A \otimes \unit_\A)
	=
	\Lambda \circ \Psi (  \unit_\A \otimes  A)
\]
for all $A \in \A ,$
where a linear map 
$\Phi \colon \A \atensor \A \to \B$
for a \cstar-algebra $\B$ is called a channel if $\Phi$ 
is unit-preserving and CP, i.e.\
\[
	\sum_{i,j=1}^n B_i^\ast \Phi (X_i^\ast X_j) B_j
	\geq 0
\]
for all $1 \leq n < \infty ,$
$\{ B_i \}_{i=1}^n \subset \B ,$
and
$\{  X_i \}_{i=1}^n \subset \A \atensor \A .$
If $\Lambda$ has a broadcasting channel 
$\Psi \in \chset{\A \otimes \A}{\A}$
(respectively, $\Psi \in \chset{\A \atensor \A}{\A}$),
then $\Lambda$ is called broadcastable in the sense of injective \cstar-tensor product
(respectively, algebraic tensor product).

The following theorem, which is immediate from Corollary~1 of 
Ref.~\onlinecite{kuramochi2018access},
is an operator algebraic version of the quantum no-broadcasting theorem
and will play an important role in the proof of the main result.
\begin{theorem}[No-broadcasting theorem for normal channels]
\label{theo:nb}
Let $\Lambda \in \nchset{\M}{\Min}$ be a normal channel.
Then the following conditions are equivalent.
\begin{enumerate}[(i)]
\item
$\Lambda$ is broadcastable in the sense of algebraic tensor product.
\item
$\Lambda$ is broadcastable in the sense of injective \cstar-tensor product.
\item
$\Lambda$ is randomization-equivalent to a QC channel with the input space $\Min .$
\end{enumerate}
\end{theorem}

For the broadcastability for channels with outcome \cstar-algebras, 
we have the following lemma.

\begin{lemma}\label{lemm:bc}
Let $\A$ be a \cstar-algebra, 
let $\Min$ be a von Neumann algebra, 
and let $\Lambda \in \chset{\A }{\Min}$
be a channel.
Suppose that $\Lambda$ has a broadcasting channel
$\Psi \in \chset{\A \otimes \A}{\A} .$
Then the normal extension $\bL \in \nchset{\A^{\ast \ast}}{\Min}$
of $\Lambda$ satisfies the conditions (i)-(iii) of Theorem~\ref{theo:nb}.
\end{lemma}
\begin{proof}
It is sufficient to show that $\bL$ satisfies the condition (i) of 
Theorem~\ref{theo:nb}.
Let $(\cH_\A , \pi_\A)$ be the universal representation of $\A .$
From the complete positivity of 
$\Psi \colon \A \otimes \A \to \A \subset \calL (\cH_\A) ,$
we can take a Stinespring representation $(\cK , \pi ,V)$ of $\Psi ,$
namely
$\cK$ is a Hilbert space, $\pi \colon \A \otimes \A \to \LK$
is a representation,
and $V \colon \cH_\A \to \cK$ is a linear isometry such that
\[
	\Psi (X)
	=
	V^\ast \pi (X) V 
	\quad
	(X \in \A \otimes \A) .
\]
We define representations $\pi_L$ and $\pi_R$ by
\[
	\pi_L (A)
	:=
	\pi (A \otimes \unit_\A) ,
	\quad
	\pi_R (A)
	:=
	\pi (\unit_\A \otimes A )
	\quad
	(A \in \A) .
\]
By the universality of the enveloping algebra $\A^{\ast \ast}$
the representations $\pi_L$ and $\pi_R$ uniquely extend to 
normal representations
$\bpi_L \colon \A^{\ast \ast} \to \LK$
and
$\bpi_R \colon \A^{\ast \ast} \to \LK ,$
respectively.
Since $\pi_L (\A)$ and $\pi_R (\A)$ commute,
so do the von Neumann algebras 
$\bpi_L (\A^{\ast \ast}) = \pi_L (\A)^{\prime \prime}$
and 
$\bpi_R (\A^{\ast \ast}) = \pi_R (\A)^{\prime \prime} .$
Thus we can define a representation
$\tpi \colon \A^{\ast \ast} \atensor \A^{\ast \ast} \to \LK$
and a channel $\tPsi \colon \A^{\ast \ast} \atensor \A^{\ast \ast} \to \A^{\ast \ast}$ by
\begin{gather*}
	\tpi \left(
	\sum_i A_i^{\prime \prime} \otimes B_i^{\prime \prime}
	\right)
	:=
	\sum_i
	\bpi_L (A_i^{\prime \prime})
	\bpi_R (B_i^{\prime \prime})
	\quad
	(A_i^{\prime \prime}, B_i^{\prime \prime} \in \A^{\ast \ast}) ,
	\\
	\tPsi (X) := V^\ast  \tpi (X) V
	\quad (X \in \A^{\ast \ast} \atensor \A^{\ast \ast} ) .
\end{gather*}
Then for each $A \in \A ,$
\[
	\bL \circ \tPsi (A \otimes \unit_{\A^{\ast \ast}})
	=
	\Lambda \circ \Psi (A \otimes \unit_\A) = \Lambda (A) = \bL (A) ,
\]
and similarly $\bL \circ \tPsi (\unit_{\A^{\ast \ast}} \otimes A) = \bL (A) .$
Since the channels
\begin{gather*}
	\A^{\ast \ast} \ni A^{\prime \prime} \longmapsto 
	\bL \circ \tPsi (A^{\prime \prime} \otimes \unit_{\A^{\ast \ast}})
	\in \Min
	\\
	\A^{\ast \ast} \ni A^{\prime \prime} \longmapsto 
	\bL \circ \tPsi (\unit_{\A^{\ast \ast}} \otimes A^{\prime \prime}  )
	\in \Min
\end{gather*}
are normal and $\A$ is $\sigma$-weakly dense in $\A^{\ast \ast} ,$
this implies that $\tPsi$ is a broadcasting channel of $\bL .$
\end{proof}

\section{Separable states} \label{sec:sep}
For \cstar-algebras $\A$ and $\B$ and for states
$\phi \in \cS (\A)$ and $\psi \in \cS (\B) ,$
there exists a unique state $\phi \otimes \psi \in \cS (\A \otimes \B)$
such that
$( \phi \otimes \psi ) (A \otimes B) = \phi(A) \psi (B) $
$(A \in \A , B \in \B) .$
For subsets $X \subset \cS(\A)$ and $Y \subset \cS (\B) ,$
we write
\[
	X \otimes Y
	:=
	\set{  
	\phi \otimes \psi \in \cS (\A \otimes \B)
	| 
	\phi \in X, \,
	\psi \in Y
	} .
\]
The states in $\cS (\A) \otimes \cS (\B)$
are called the product states.
We define 
\[
	\Ssep (\A \otimes \B)
	:=
	\cco ( \cS (\A) \otimes \cS (\B) ) ,
\]
where the closed convex hull is with respect to the weak-$\ast$
topology $\sigma ((\A \otimes \B)^\ast , \A \otimes \B) .$
Each element of $\Ssep (\A \otimes \B)$ is called a separable state,
while each element of the complement 
$\cS (\A \otimes \B) \setminus \Ssep (\A \otimes \B)$
is called an entangled state.

\begin{lemma}\label{lemm:prodch}
Let $\Lambda \in \chset{\A}{\B}$
and 
$\Gamma \in \chset{\C}{\D}$
be channels.
Then $\omega \circ (\Lambda \otimes \Gamma) \in \Ssep (\A \otimes \C)$
for any $\omega \in \Ssep (\B \otimes \D) .$
\end{lemma}
\begin{proof}
The claim immediately follows from that
$\Lambda \otimes \Gamma$
maps any input product state $\phi \otimes \psi \in \cS (\B) \otimes \cS (\D)$
to the product state
$(\phi \circ \Lambda) \otimes ( \psi \circ  \Gamma) \in \cS (\A) \otimes  \cS (\C) .$
\end{proof}

\begin{lemma}\label{lemm:comm}
Let $\A$ and $\B$ be \cstar-algebras.
Suppose that $\A$ is commutative.
Then 
$\cS (\A \otimes \B) = \Ssep (\A \otimes \B) .$
\end{lemma}
\begin{proof}
By Theorem~4.4 of Ref.~\onlinecite{takesakivol1}, 
we have 
$\Sp (\A \otimes \B) = \Sp (\A) \otimes \Sp (\B) .$
Hence by the Krein-Milman theorem,
$\cS (\A \otimes \B) = \cco (  \Sp (\A) \otimes \Sp (\B)  )
\subset
\cco (  \cS (\A) \otimes \cS(\B)  )
=
\Ssep (\A \otimes \B)
\subset 
\cS (\A \otimes \B) ,
$
which proves the claim.
\end{proof}

\begin{lemma}
\label{lemm:iso}
Let $\A$ and $\B$ be \cstar-algebras.
Then the map
\[
	f \colon \cS (\A) \times \cS (\B)
	\ni
	(\phi , \psi)
	\longmapsto
	\phi \otimes \psi 
	\in 
	\cS (\A) \otimes \cS (\B)
\]
is a topological isomorphism,
where the topologies of 
$ \cS (\A) \times \cS (\B)$
and
$ \cS (\A) \otimes \cS (\B)$
are defined as the product topology of the weak-$\ast$ topologies of 
$\cS (\A)$ and $\cS(\B) ,$
and the weak-$\ast$ topology on $\cS (\A \otimes \B) ,$
respectively.
\end{lemma}
\begin{proof}
We can easily see that $f$ is bijective.
Since $\cS (\A) \times \cS (\B)$ is a compact Hausdorff space,
we have only to prove the continuity of $f .$
Let $(\phi_i , \psi_i)_{i \in I}$ be a net on $\cS (\A) \times \cS (\B)$
converging to $(\phi , \psi) \in \cS (\A) \times \cS (\B) .$
Then $(\phi_i \otimes \psi_i) (X) \to (\phi \otimes \psi) (X)$
for all $X \in \A \atensor \B .$
Since the net $(\phi_i \otimes \psi_i)_{i \in I}$
is uniformly bounded
and $\A \atensor \B$ is norm dense in $\A \otimes \B ,$
this implies
$(\phi_i \otimes \psi_i) (X) \to (\phi \otimes \psi) (X)$
for all $X \in \A \otimes \B ,$
which proves the continuity of $f .$
\end{proof}

\begin{corollary}\label{coro:compact}
For \cstar-algebras $\A $ and $\B ,$
the set $\cS (\A) \otimes \cS(\B)$ of product states 
is a weakly-$\ast$ compact subset of $\cS (\A \otimes \B) .$
\end{corollary}

To prove the integral representation for separable states,
we need some results of barycentric integrals on compact subsets of general 
locally convex Hausdorff spaces.~\cite{phelps2001lectures}
Let $E$ be a locally convex Hausdorff space and let $K \subset E$
be a compact subset.
A regular probability measure $ \mu \in \bfMpo (K)$
is said to represent a point $\bar{x} \in E$ if
\[
	f (\bar{x}) = \int_K f(x) d\mu (x)
\]
for any continuous linear functional $f \in E^\ast .$
By using the Hahn-Banach separation theorem, 
we can show that such a point $\bar{x} $ is,
if exists,
unique for given $\mu \in \bfMpo (K)$
and called the barycenter of $\mu .$
For $\mu \in \bfMpo (K) ,$
we denote the barycenter of $\mu $
by
\[
	\int_K x d \mu (x) ,
\]
if it exists.

\begin{proposition}[Ref.~\onlinecite{phelps2001lectures}, Proposition~1.2]
\label{prop:barycenter}
Let $E$ be a locally convex Hausdorff space and let 
$K \subset E$ be a compact subset.
Then for a point $x \in E ,$
$x \in \cco (K)$ if and only if there exists a regular probability measure
$\mu \in \bfMpo (K)$ that represents $x .$
\end{proposition}

From Corollary~\ref{coro:compact} and Proposition~\ref{prop:barycenter}
we obtain

\begin{corollary}\label{coro:sep}
Let $\A$ and $\B$ be \cstar-algebras.
Then for any separable state $\omega \in \Ssep (\A \otimes \B)$
there exists a regular probability measure 
$\nu \in \bfMpo (\cS (\A) \otimes \cS (\B))$
such that
\[
	\omega
	=
	\int_{\cS (\A) \otimes \cS (\B)}
	\phi \otimes \psi \,
	d\nu 
	(\phi \otimes \psi ) .
\]
\end{corollary}
Corollary~\ref{coro:sep} is a \cstar-algebra version of the integral representation 
for infinite-dimensional separable density operators
obtained in Ref.~\onlinecite{holevo2005separability}.

In the case of state space on a \cstar-algebra, 
the barycentric integral can be generalized to the weakly-$\ast$ measurable 
state-valued maps as follows.
Let $\A $ be a \cstar-algebra and let $(\Omega , \Sigma , \mu)$
be a probability space.
A map $\Omega \ni x \mapsto \phi_x \in \cS (\A)$ is called 
weakly-$\ast$ measurable if 
$\Omega \ni x \mapsto \phi_x (A) \in \cmplx$
is $\Sigma$-measurable for all $A \in \A .$
For such a map $\Omega \ni x \mapsto \phi_x \in \cS (\A)$
we define the integral state 
$\int_\Omega \phi_x d\mu (x) \in \cS (\A)$ by
\[
	\left(
	\int_\Omega \phi_x d\mu (x)
	\right)
	(A)
	:=
	\int_\Omega
	\phi_x (A) d\mu (x)
	\quad
	(A\in \A) .
\]

\section{Characterizations of EB channels} \label{sec:main}
In this section, we give equivalent characterizations 
of the EB condition for a channel (Theorem~\ref{theo:main}).

\begin{definition}\label{defi:eb}
A channel $\Lambda \in \chset{\A}{\Ain}$ is called EB if
$\omega \circ (\Lambda \otimes \id_\B) \in \Ssep (\A \otimes \B)$
for any \cstar-algebra $\B$ and 
any input state $\omega \in \cS (\Ain \otimes \B) .$
We denote by $\chseteb{\A}{\Ain}$ the set of EB channels 
from $\A$ to $\Ain .$
\end{definition}
In what follows in this section, we fix an input Hilbert space $\Hin ,$
which may be separable or non-separable,
and an orthonormal basis $(\xi_i)_{i \in I}$ of $\Hin .$
We write the set of finite subsets of $I$ as $\FI ,$
which is directed by the set inclusion,
and for each $F \in \FI ,$ define $P_F, $ $\cH_F ,$
$\eta_F \in \cH_F \otimes \cH_F ,$
and $\omega_F \in \cS (\LHin \otimes \calL (\cH_F))$ by
\begin{gather*}
	P_F :=
	\sum_{i \in F}
	\ket{\xi_i} \bra{\xi_i} ,
	\\
	\cH_F := P_F \Hin ,
	\\
	\eta_F
	:=
	|F|^{-1/2}
	\sum_{i \in F}
	\xi_i \otimes \xi_i ,
	\\
	\omega_F (A)
	:=
	\braket{\eta_F | A \eta_F}
	=
	|F|^{-1}
	\sum_{i,j\in F}
	\braket{\xi_i \otimes \xi_i | A (\xi_j \otimes \xi_j)}
	\quad
	(A \in \LHin \otimes \calL (\cH_F)) ,
\end{gather*}
where $| F |$ denotes the number of elements of $F .$
Since $\omega_F (P_F \otimes P_F) = 1 ,$
we can regard $\omega_F$ as a state in 
$\cS (\calL (\cH_F) \otimes \calL (\cH_F) ) ,$
which we also write as $\omega_F .$
The state $\omega_F$ is a maximally entangled state on $\cH_F \otimes \cH_F .$

Now we are in a position to prove the main result of this paper:
\begin{theorem}\label{theo:main}
Let $\Ain$ be a unital \cstar-subalgebra of $\LHin ,$
let $\A$ be a \cstar-algebra,
let $\Lambda \in \chset{\A}{\Ain} \subset \chset{\A}{\LHin}$ 
be a channel,
and let $\bL \in \nchset{\A^{\ast \ast}}{\LHin}$ be the normal extension of 
$\Lambda .$
Then the following conditions are equivalent.
\begin{enumerate}[(i)]
\item \label{it:eb1}
$\Lambda \in \chseteb{\A}{\Ain} .$
\item \label{it:eb2}
$\Lambda \in \chseteb{\A}{\LHin} .$
\item \label{it:eb3}
$\bL \in \chseteb{\A^{\ast \ast}}{\LHin} .$
\item \label{it:bell}
$\omega_F \circ (\Lambda \otimes \id_{\LHF}) \in \Ssep (\A \otimes \LHF)$
for all $F \in \FI .$
\item \label{it:rand1}
There exists a channel $\Gamma \in \chset{\C}{\LHin}$
with a commutative outcome \cstar-algebra $\C$ 
such that $\Lambda \cocp \Gamma .$
\item \label{it:rand2}
There exists a QC channel 
$\tG \in \nchset{\M}{\LHin}$ 
with a commutative outcome von Neumann algebra $\M$ such that
$\Lambda \cocp \tG .$
\item \label{it:holevo}
(Measurement-prepare form).
There exists a POVM $(\Omega , \Sigma , \oM)$
on $\Hin $ and a weakly-$\ast$ measurable map
$\Omega \ni x \mapsto \phi_x \in \cS (\A)$ 
such that
\begin{equation}
	\Lambda (A)
	=
	\int_\Omega \phi_x (A) d \oM (x)
	\quad
	(A \in \A) .
	\label{eq:holevo}
\end{equation}
\item \label{it:njoint}
$\Lambda$ has an $n$-joint channel 
$\Theta_n \in \chset{\A^{\otimes n}}{\LHin}$
for all $2 \leq n < \infty .$
\item \label{it:ijoints}
$\Lambda$ has a symmetric $\omega$-joint channel
$\Theta_\infty \in \chset{\A^{\otimes \infty}}{\LHin} .$
\item \label{it:ijoint}
$\Lambda$ has an $\omega$-joint channel
$\Theta_\infty^\prime \in \chset{\A^{\otimes \infty}}{\LHin} .$
\end{enumerate}
\end{theorem}
\begin{proof}
\eqref{it:eb1}$\implies$\eqref{it:eb2}.
Assume \eqref{it:eb1} and take a \cstar-algebra $\B$
and a state $\omega \in \cS (\LHin \otimes \B) .$
Since $\Ain \otimes \B \subset \LHin \otimes \B ,$
by the assumption we have
$\omega \circ (\Lambda \otimes \id_\B) 
= \omega \rvert_{\Ain \otimes \B} \circ (\Lambda \otimes \id_\B) 
\in \Ssep (\A \otimes \B) ,
$
which implies \eqref{it:eb2}.

\eqref{it:eb2}$\implies$\eqref{it:eb1}.
Assume \eqref{it:eb2} and take a \cstar-algebra $\B$
and a state $\omega \in \cS (\Ain \otimes \B) .$
From $\Ain \otimes \B \subset \LHin \otimes \B $
and the Hahn-Banach theorem,
$\omega$ extends to a linear functional $\tomega \in (\LHin \otimes \B)^\ast$
such that $\| \tomega \| = \| \omega \| =1 .$
Since 
$\tomega (\unit_{\LHin \otimes \B}) 
=  \omega( \unit_{\Ain \otimes \B})
= 1 = \| \tomega \| ,$
$\tomega$ is a state on $\LHin \otimes \B $
by Proposition~1.5.2 of Ref.~\onlinecite{sakaibook}.
Thus by the assumption we have
$\omega \circ (\Lambda \otimes \id_\B) = \tomega \circ (\Lambda \otimes \id_\B)
\in \Ssep (\A \otimes \B),
$
which implies \eqref{it:eb1}.

\eqref{it:eb2}$\implies$\eqref{it:bell} is obvious.

\eqref{it:bell}$\implies$\eqref{it:rand1}.
Assume \eqref{it:bell}.
For each $F \in \FI ,$ we define $\LF \in \chset{\A}{\LHF}$ by
$\LF (A) := P_F \Lambda (A) P_F $
$(A\in \A ) .$
By assumption we have
$\omega_F \circ (\Lambda_F \otimes \id_{\LHF})
= \omega_F \circ (\Lambda \otimes \id_{\LHF})
\in \Ssep (\A \otimes \LHF) .
$
Therefore by Corollary~\ref{coro:sep} there exists a regular probability measure
$\nu_F \in \bfMpo (\cS (\A) \otimes \cS (\LHF))$
such that
\[
	\omega_F \circ (\Lambda_F \otimes \id_{\LHF})
	=
	\int_{\cS (\A) \otimes \cS (\LHF)}
	\phi \otimes \psi \,
	d\nu_F (\phi \otimes \psi).
\]
For each Borel set $E \in \B (\SA),$ 
we define a positive linear functional 
$\tphi_E \in \LHF^\ast$ by
\[
	\tphi_E (B)
	:=
	\int_{E \otimes \cS (\LHF )}
	\psi(B)
	d\nu_F 
	(\phi \otimes \psi)
	\quad
	(B \in \LHF).
\]
We write $\tphi := \tphi_{\SA} ,$ 
which is a state on $\LHF$
given by
\begin{equation}
	\tphi (B)
	=
	\omega_F
	\circ
	(\Lambda_F \otimes \id_{\LHF})
	(\unit_{\A} \otimes B )
	=
	\braket{\eta_F | (P_F \otimes B) \eta_F} 
	\quad (B \in \LHF) .
	\label{eq:tphi}
\end{equation}
Then the RHS of \eqref{eq:tphi}
is a GNS representation of $\tphi .$
Furthermore, by the $\sigma$-additivity of $\nu_F ,$
the map 
$\B (\SA) \ni E \mapsto \tphi_E (B) \in \cmplx$
is also $\sigma$-additive for all $B \in \LHF .$
Therefore application of the Radon-Nikodym theorem 
for CP maps (e.g.\ Theorem~III.1 of Ref.~\onlinecite{doi:10.1063/1.1615697})
to $\tphi \in \chset{\LHF}{\cmplx} $
yields that
there exists a unique POVM $(\SA , \B (\SA) , \oM_F)$
on $\cH_F$ such that
\[
	\tphi_E (B)
	=
	\braket{\eta_F | (\oM_F (E) \otimes B) \eta_F}
	\quad
	(E \in \B (\SA) , B \in \LHF  ) .
\]
We define channels 
$\alpha \in \chset{\A}{C (\SA)}$
and
$\Gamma_F \in \chset{C ( \SA )}{\LHF}$
by
\begin{gather*}
	\alpha (A) (\phi)
	:=
	\phi (A)
	\quad
	(A \in \A , \phi \in \SA) ,
	\\
	\Gamma_F (f)
	:=
	\int_{\SA}
	f (\phi)
	d \oM_F (\phi)
	\quad
	(f \in C (\SA ) ) .
\end{gather*}
Then for each $i,j \in F$
and each 
$E \in \B (\SA ) ,$
\[
	\braket{\eta_F| 
	(\oM_F (E) \otimes \ket{\xi_i} \bra{\xi_j})
	\eta_F
	}
	=
	\tphi_E (\ket{\xi_i } \bra{\xi_j})
	=
	\int_{E \otimes \cS (\LHF)}
	\psi (\ket{\xi_i } \bra{\xi_j})
	d\nu_F (\phi \otimes \psi) .
\]
Hence for each $i, j \in F $
and each $A \in \A  ,$
\begin{align*}
	\braket{\xi_i | \Gamma_F \circ \alpha (A) \xi_j}
	&=
	\int_{\SA}
	\phi(A) 
	d \braket{\xi_i| \oM_F (\phi) \xi_j}
	\\
	&=
	| F |
	\int_{\SA}
	\phi (A)
	d \braket{\eta_F| 
	(\oM_F (\phi) \otimes \ket{\xi_i} \bra{\xi_j})
	\eta_F
	}
	\\
	&= |F|
	\int_{\SA \otimes \cS (\LHF )}
	\phi (A) 
	\psi(\ket{\xi_i} \bra{\xi_j})
	d \nu_F (\phi \otimes \psi)
	\\
	&=
	|F|
	\left(
	\omega_F \circ (\Lambda_F \otimes \id_{\LHF})
	\right)
	(A \otimes \ket{\xi_i} \bra{\xi_j} )
	\\
	&= 
	\braket{\xi_i | \Lambda_F (A) \xi_j} .
\end{align*}
Therefore 
$\Lambda_F = \Gamma_F \circ \alpha .$
Since a closed ball in $\LHin$ is $\sigma$-weakly compact,
Tychonoff\rq{}s theorem implies that
there exist
a subnet $(\Gamma_{F(j)})_{j \in J}$
and a bounded linear map 
$\Gamma \colon C (\SA) \to \LHin$
such that
$\Gamma_{F(j)} (f) 
\xrightarrow{\text{$\sigma$-weakly}}
\Gamma (f)
$
for all $f \in C (\SA ) .$
By the complete positivity of each $\Gamma_F ,$
$\Gamma $ is also CP.
From
\[
	\Gamma (\unit_{C (\SA )})
	=
	\swlim_{j \in J} \Gamma_{F(j)} 
	(\unit_{C (\SA )})
	=
	\swlim_{j \in J} P_{F(j)}
	= \unit_{\Hin} ,
\]
where $\swlim$ denotes the $\sigma$-weak limit,
we have 
$\Gamma \in \chset{C(\SA ) }{\LHin} .$
Moreover for each $A \in \A $
\[
	\Gamma \circ \alpha (A)
	=
	\swlim_{j \in J}
	\Gamma_{F(j)} \circ \alpha (A)
	=
	\swlim_{j \in J}
	\Lambda_{F(j)} (A)
	= \Lambda (A),
\]
where we used 
$P_F B P_F \xrightarrow{\text{$\sigma$-strongly}}
B$
$(B \in \LHin)$
in the third equality.
Therefore 
$\Lambda = \Gamma \circ \alpha \cocp \Gamma .$
Since the outcome space $C (\SA ) $
of $\Gamma$ is commutative, 
this implies \eqref{it:rand1}.

\eqref{it:rand1}$\implies$\eqref{it:rand2}.
Let $\Gamma \in \chset{\C}{\LHin}$ be a channel 
such that $\C$ is commutative and 
$\Lambda \cocp \Gamma .$
Then the outcome space $\C^{\ast \ast}$
of the normal extension 
$\tG \in \nchset{\C^{\ast \ast}}{\LHin}$
of $\Gamma $
is also commutative
and hence $\tG$ is a QC channel.
Moreover we have
$\Lambda \cocp \Gamma \cocp \tG ,$
which implies \eqref{it:rand2}.

\eqref{it:rand2}$\implies$\eqref{it:rand1} is obvious.

\eqref{it:rand1}$\implies$\eqref{it:holevo}.
Assume \eqref{it:rand1} and take channels
$\Gamma \in \chset{\C}{\LHin}$
and
$\beta \in \chset{\A}{\C}$
such that $\C$ is commutative 
and $\Lambda = \Gamma \circ \beta .$
By the Gelfand representation we may assume that
$\C = C(\Omega)$ for some compact Hausdorff space $\Omega .$
Then by Proposition~\ref{prop:RMK} there exists a unique
regular POVM $(\Omega , \B (\Omega) , \oM)$
on $\Hin$ such that
\[
	\Gamma (f)
	= \int_\Omega f (x) d\oM (x)
	\quad
	(f \in C (\Omega)) .
\]
If we define $\phi_x \in \SA$ by $\phi_x (A) := \beta(A) (x) $
$(A \in \A)$
for each $x \in \Omega ,$
the map $\Omega \ni x \mapsto \phi_x \in \SA$
is weakly-$\ast$ measurable.
Moreover for each $A \in \A$
\[
	\Lambda (A)
	=
	\Gamma \circ \beta (A)
	=
	\int_\Omega \beta (A) (x) d \oM (x)
	=
	\int_\Omega \phi_x (A) d \oM (x),
\]
which implies \eqref{it:holevo}.

\eqref{it:holevo}$\implies$\eqref{it:rand1}.
Assume \eqref{it:holevo} and take a POVM $(\Omega , \Sigma , \oM)$ on $\Hin$
and a weakly-$\ast$ measurable map 
$\Omega \ni x \mapsto \phi_x \in \SA $
satisfying \eqref{eq:holevo}.
We define channels $\gamma \in \chset{\A}{B(\Omega , \Sigma)}$
and $\Gamma^\prime \in \chset{B(\Omega , \Sigma)}{\LHin}$
by
\begin{gather*}
	\gamma (A) (x)
	:=
	\phi_x (A)
	\quad
	(A \in \A , x \in \Omega),
	\\
	\Gamma^\prime (f)
	:=
	\int_\Omega 
	f (x) d \oM (x)
	\quad
	(f \in B(\Omega , \Sigma)) .
\end{gather*}
Then by the assumption~\eqref{eq:holevo}
we have $\Lambda = \Gamma^\prime \circ \gamma \cocp \Gamma^\prime .$
Since $B(\Omega, \Sigma)$ is commutative, 
this implies \eqref{it:rand1}.

\eqref{it:rand1}$\implies$\eqref{it:eb2}.
Assume \eqref{it:rand1} and take channels
$\Gamma \in \chset{\C}{\LHin}$
and
$\alpha \in \chset{\A}{\C}$
such that $\C$ is commutative 
and $\Lambda = \Gamma \circ \alpha .$
Then for any \cstar-algebra $\B$
and any state $\omega \in \cS (\LHin \otimes \B) ,$
Lemma~\ref{lemm:comm} implies
$\omega \circ (\Gamma \otimes \id_\B) \in \Ssep (\C \otimes \B) .$
Hence by Lemma~\ref{lemm:prodch}
\[
	\omega \circ (\Lambda \otimes \id_\B)
	=
	\omega \circ 
	(\Gamma \otimes \id_\B)
	\circ
	(\alpha \otimes \id_\B)
	\in \Ssep (\A \otimes \B) ,
\]
which proves \eqref{it:eb2}.

\eqref{it:rand1}$\implies$\eqref{it:njoint}.
Assume \eqref{it:rand1} and take channels
$\Gamma \in \chset{\C}{\LHin}$
and
$\alpha \in \chset{\A}{\C}$
such that $\C$ is commutative 
and $\Lambda = \Gamma \circ \alpha .$
By the commutativity of $\C , $
for each $2 \leq n < \infty $
there exists a channel
$\alpha_n \in \chset{\Atn}{\C}$ such that
\[
	\alpha_n (A_1 \otimes A_2 \otimes \dots \otimes A_n)
	=
	\alpha (A_1) \alpha (A_2) \cdots \alpha (A_n)
	\quad
	(A_1 , A_2 , \dots , A_n \in \A) .
\]
Then $\Gamma \circ \alpha_n \in \chset{\Atn}{\LHin}$
is an $n$-joint channel of $\Lambda .$

\eqref{it:njoint}$\implies$\eqref{it:ijoints}.
Assume \eqref{it:njoint} and take a symmetric $n$-joint channel
$\Theta_n \in \chset{\Atn}{\LHin}$
of $\Lambda$ for each $2 \leq n < \infty .$
We write $\tA_0 := \bigcup_{n \geq 1} \Atn \otimes \unit_{\Ainf} ,$
which is a norm dense $\ast$-subalgebra of $\Ainf ,$
and for each $2\leq n < \infty$
define a map $\Xi_n \colon \tA_0 \to \LHin$
by
\[
	\Xi_n (Y)
	:=
	\begin{cases}
	\Theta_n (X)
	&
	\text{if $Y = X \otimes \unit_{\Ainf}$ and $ X \in \Atn  ;$}
	\\
	0
	&
	\text{otherwise .}
	\end{cases}
\]
Since $\| \Xi_n (Y) \| \leq  \| Y \|$ for all $Y \in \tA_0$
and all $2 \leq n < \infty ,$
by Tychonoff\rq{}s theorem there exists a subnet
$(\Xi_{n(k)})_{k \in K}$
such that the $\sigma$-weak limit 
$\Theta_0 (Y) := \swlim_{k \in K} \Xi_{n(k)} (Y)$
exists for all $Y \in \tA_0 .$
Since $\Xi_{n(k)} (c_1 Y_1 + c_2 Y_2) = c_1 \Xi_{n(k)} (Y_1) + c_2 \Xi_{n(k)} (Y_2)$
eventually for each $Y_1 , Y_2 \in \tA_0$
and each $c_1 , c_2 \in \cmplx ,$
$\Theta_0 \colon \tA_0 \to \LHin$ is a bounded linear map
and hence uniquely extends to a bounded linear map
$\Theta_\infty \colon \Ainf \to \LHin .$
Then,
by the complete positivity of each $\Theta_n ,$
$\Theta_\infty $ is a channel in
$\chset{\Ainf}{\LHin} .$
By the symmetry of each $\Theta_n ,$
$\Theta_\infty$ is also symmetric.
Moreover
\begin{equation}
	\Theta_\infty
	(A \otimes \unit_{\Ainf})
	=
	\swlim_{k \in K}
	\Theta_{n(k)}
	(A \otimes \unit_{\A^{\otimes (n(k)-1)}})
	=\Lambda (A)
	\quad
	(A \in \A) .
	\label{eq:thinf}
\end{equation}
From the symmetry of $\Theta_\infty ,$
Eq.~\eqref{eq:thinf}
implies that $\Theta_\infty$
is a symmetric $\omega$-joint channel of $\Lambda, $
which proves \eqref{it:ijoints}.

\eqref{it:ijoints}$\implies$\eqref{it:ijoint} is obvious.

\eqref{it:ijoint}$\implies$\eqref{it:njoint} is immediate from that 
for an $\omega$-joint channel $\Theta_\infty^\prime \in \chset{\Ainf}{\LHin}$
and each $2 \leq n< \infty ,$
the map 
\[
	\Atn \ni X \longmapsto \Theta_\infty^\prime (X \otimes \unit_{\Ainf})
	\in \LHin
\]
is an $n$-joint channel of $\Lambda .$

\eqref{it:ijoints}$\implies$\eqref{it:rand2}.
Let $\Theta_\infty \in \chset{\Ainf}{\LHin}$
be a symmetric $\omega$-joint channel of $\Lambda .$
If $\Theta_\infty$ has a broadcasting channel 
$\Psi_\infty \in \chset{\Ainf \otimes \Ainf}{\Ainf} ,$
by Lemma~\ref{lemm:bc}
the normal extension $\bTh_\infty \in \chset{(\Ainf)^{\ast \ast}}{\LHin}$
of $\Theta_\infty$
is randomization-equivalent to a QC channel
and, from $\Lambda \cocp \Theta_\infty \cocp \bTh_\infty ,$
this implies \eqref{it:rand2}.
Thus it suffices to construct a broadcasting channel of $\Theta_\infty .$
For each $1 \leq n < \infty$
we define a representation 
\[
\rho_n \colon (\Atn \otimes \unit_{\Ainf}) \otimes (\Atn \otimes \unit_{\Ainf})
\to
\A^{\otimes (2n)} \otimes \unit_{\Ainf}
\]
by
\begin{align*}
	&\rho_n \left(
	(A_1 \otimes A_2 \otimes \cdots \otimes A_n \otimes \unit_{\Ainf})
	\otimes
	(B_1 \otimes B_2 \otimes \cdots \otimes B_n \otimes \unit_{\Ainf})
	\right)
	\\
	&:=
	A_1 \otimes B_1 \otimes A_2 \otimes B_2 \otimes
	\cdots 
	A_n \otimes B_n \otimes \unit_{\Ainf}
	\\
	&( A_1 , \dots , A_n , B_1 , \dots , B_n \in \A) .
\end{align*}
Let $\B_n :=  (\Atn \otimes \unit_{\Ainf}) \otimes (\Atn \otimes \unit_{\Ainf})$
and $\B_\infty := \bigcup_{1 \leq n < \infty} \B_n ,$
which are \cstar- and $\ast$-subalgebras of $\Ainf \otimes \Ainf, $
respectively.
Since $\rho_n \rvert_{\B_m} = \rho_m$
for $1 \leq m < n < \infty ,$
there exists a representation
$\trho_0 \colon \B_\infty \to \Ainf$
such that $\trho_0 \rvert_{\B_n} = \rho_n$
for all $1 \leq n < \infty .$
Now we show that $\B_\infty$ is norm dense in $\Ainf \otimes \Ainf .$
Take an arbitrary product element $X_1 \otimes X_2 \in \Ainf \otimes \Ainf$
$(X_1 , X_2 \in \Ainf) .$
Then there exist sequences $(X_{j,m})_{m \geq 1}$
$(j=1,2)$
in $\bigcup_{n \geq 1} \Atn \otimes \unit_{\Ainf}$
such that
$ \| X_{j,m} - X_j \| \to 0  $
and 
$\sup_{m \geq 1} \| X_{j,m} \| \leq \| X_j \| .$
Then $(X_{1,m} \otimes X_{2,m})_{m \geq 1}$
is a sequence in $\B_\infty$ and satisfies
\begin{align*}
	\| X_{1,m} \otimes X_{2,m} - X_{1} \otimes X_{2} \|
	&\leq 
	\|  (X_{1,m} - X_1) \otimes X_{2,m} \|
	+
	\| X_1 \otimes (X_{2,m} - X_2) \|
	\\
	&\leq
	\| X_{1,m} - X_1 \| \| X_2 \|
	+ 
	\| X_1 \| 
	\| X_{2,m} - X_2 \| 
	\\
	&\to 0 .
\end{align*}
Therefore $\B_\infty$ is norm dense in $\Ainf \otimes \Ainf .$
Hence $\trho_0$ uniquely extends to a representation
$\trho \colon \Ainf \otimes \Ainf \to \Ainf .$
Then by the symmetry of $\Theta_\infty ,$
for each $2 \leq n < \infty$
and each $X \in \Atn ,$
\[
	\Theta_{\infty} (X \otimes \unit_{\Ainf})
	=
	\Theta_\infty \circ \trho 
	\left(
	(X \otimes \unit_{\Ainf}) \otimes \unit_{\Ainf}
	\right)
	=
	\Theta_\infty \circ \trho 
	\left(
	 \unit_{\Ainf} \otimes  (X \otimes \unit_{\Ainf})
	\right) .
\]
Since $\bigcup_{n \geq 1} \Atn \otimes \unit_{\Ainf}$
is norm dense in $\Ainf ,$
this implies that $\trho$ is a broadcasting channel of $\Theta_\infty ,$
which proves \eqref{it:rand2}.

\eqref{it:eb2}$\iff$\eqref{it:eb3} follows from the equivalence
\eqref{it:eb2}$\iff$\eqref{it:rand2}
and from that $\Lambda \cocp \tG$ if and only if
$\bL \cocp \tG$ for any QC channel $\tG .$
\end{proof}
From the proof of \eqref{it:ijoints}$\implies$\eqref{it:rand2}
in Theorem~\ref{theo:main},
we obtain
\begin{corollary}\label{coro:sym}
Let $\A$ be a \cstar-algebra.
Then any symmetric channel
$\Theta \in \chset{\Ainf}{\LHin}$
has a broadcasting channel $\Psi \in \chset{\Ainf \otimes \Ainf}{\Ainf} .$
Furthermore the normal extension
$\bTh \in \nchset{(\Ainf)^{\ast \ast}}{\LHin}$
of $\Theta$ is randomization-equivalent to a QC channel.
\end{corollary}

\begin{remark} \label{rem:refs}
Related results to the equivalence of the EB condition and \eqref{it:rand1} 
of Theorem~\ref{theo:main} were obtained in 
Ref.~\onlinecite	{STORMER20082303} (Corollary~3)
and
Ref.~\onlinecite{doi:10.1063/1.5024385} (Lemma~III.1).
\end{remark}

\section{Dedekind-closedness of EB channels} \label{sec:dedekind}
In this section we prove that the supremum or infimum
of any randomization-monotone net of normal EB channels
with a fixed input von Neumann algebra is also EB
(Theorem~\ref{theo:dc}).
Before going to Theorem~\ref{theo:dc}, 
we need to review some results from Ref.~\onlinecite{kuramochi2018directed}.

For a von Neumann algebra $\Min$
we denote by $\nchset{}{\Min}$
the class of normal channels with the input space $\Min$
and with arbitrary outcome von Neumann algebras.
Since the class of von Neumann algebras is a proper class,
so is $\nchset{}{\Min} .$

\begin{proposition}[Ref.~\onlinecite{kuramochi2018directed}, Section~3]
\label{prop:class}
Let $\Min$ be a von Neumann algebra.
Then there exist a set $\CH (\Min)$
and a class-to-set surjection
\begin{equation}
	\nchset{}{\Min} \ni \Lambda \longmapsto [\Lambda] \in \CH (\Min)
	\label{eq:csmap}
\end{equation}
such that for any $\Lambda , \Gamma \in \nchset{}{\Min} ,$
$\Lambda \eqcp \Gamma$ if and only if $[\Lambda] = [\Gamma] .$
\end{proposition}

For each von Neumann algebra $\Min$ we fix such a set
$\CH (\Min)$ and a map \eqref{eq:csmap}.
We call $\CH (\Min)$ the set of randomization-equivalence classes of normal channels.
We define a partial order $\cocp $ on $\CH (\Min)$ by
$[\Lambda] \cocp [\Gamma] $
$:\defarrow$
$\Lambda \cocp \Gamma $
$([\Lambda ] , [\Gamma] \in \CH (\Min)) .$
We also define 
\begin{gather*}
	\CHqc (\Min)
	:=
	\set{ [\Lambda] \in \CH (\Min) | 
	\text{the outcome algebra of $\Lambda$ is commutative}} ,
	\\
	\CHeb (\Min ):=
	\set{ [\Lambda] \in \CH (\Min) | 
	\text{$\Lambda$ is EB}
	 },
\end{gather*}
which are the sets of equivalence classes of QC and EB normal channels,
respectively.

Let $(X , \leq )$ be a partially ordered set (poset).
We adopt the following terminology as in Ref.~\onlinecite{kuramochi2018directed}
\begin{itemize}
\item
A net $(x_i)_{i \in I}$ on $X$ is called increasing
(respectively, decreasing)
if for $i,j \in I ,$
$i \leq j$ implies $x_i \leq x_j$
(respectively, $x_j \leq x_i$).
\item
$X$ is called an upper (respectively, lower) 
directed-complete partially ordered set (dcpo)
if any increasing (respectively, decreasing) net $(x_i)_{i \in I}$
on $X$ has a supremum $\sup_{i \in I} x_i \in X$
(respectively, infimum $\inf_{i \in I} x_i \in X $).
\item
A subset $A \subset X$ is said to be an upper (respectively, lower) 
Dedekind-closed subset of $X$
if whenever an increasing (respectively, decreasing)
net $(x_i)_{i \in I}$ on $A$ 
has a supremum $\sup_{i \in } x_i \in X $
(respectively, infimum $\inf_{i \in I }x_i \in X$),
then $\sup_{i\in I} x_i \in A$
(respectively, $\inf_{i \in I}x_i \in A$).
\end{itemize}

The following theorem is the \lq\lq{}channel part\rq\rq{}
of the main results of Ref.~\onlinecite{kuramochi2018directed}.
\begin{theorem}[Ref.~\onlinecite{kuramochi2018directed}, Theorem~3]
\label{theo:prev}
Let $\Min$ be a von Neumann algebra. 
\begin{enumerate}[(i)]
\item
$\CH (\Min)$ is an upper and lower dcpo.
\item
$\CHqc (\Min)$ is an upper and lower Dedekind-closed subset of $\CH (\Min) .$
\end{enumerate}
\end{theorem}

Now we are in a position to prove the following theorem.

\begin{theorem}\label{theo:dc}
Let $\Min$ be a von Neumann algebra.
Then $\CHeb (\Min)$ is an upper and lower Dedekind-closed subset 
of $\CH (\Min) .$
\end{theorem}
\begin{proof}
Suppose that $\Min$ acts on a Hilbert space $\Hin .$
We first show the claim when $\Min = \LHin .$

By Theorem~\ref{theo:main}, we can write
\[
	\CHeb (\LHin)
	=
	\set{[\Lambda] \in \CH (\LHin)| 
	\exists [\Gamma]\in \CHqc (\LHin) \text{ s.t. } 
	[\Lambda] \cocp [\Gamma]
	} ,
\]
from which the lower Dedekind-closedness of $\CHeb (\LHin)$ immediately follows.
Now we show the upper Dedekind-closedness of $\CHeb (\LHin) .$
According to Lemma~4 of Ref.~\onlinecite{kuramochi2018directed},
we have only to prove
$\sup_{\alpha< \alpha_0} [\Lambda_\alpha] \in \CHeb (\LHin)$
for any increasing transfinite sequence $([\Lambda_\alpha])_{\alpha <\alpha_0}$
in $\CHeb (\LHin) $
indexed by ordinals $\alpha$ 
smaller than $\alpha_0$ and greater than or equal to $0 .$
From the proof of Lemma~7 of Ref.~\onlinecite{kuramochi2018directed}
we can take channels $\tL_\alpha \in \chset{\A_\alpha}{\LHin}$
$(\alpha < \alpha_0)$
and 
$\Lambda \in \chset{\A}{\LHin}$ such that
\begin{itemize}
\item
$\Lambda_\alpha \eqcp \tL_\alpha $ 
$(\alpha < \alpha_0) ;$
\item
$(\A_\alpha)_{\alpha < \alpha_0}$ is an increasing transfinite sequence of 
unital \cstar-subalgebras of $\A$
such that $\bigcup_{\alpha < \alpha_0} \A_\alpha$
is a norm dense $\ast$-subalgebra of $\A ;$
\item
$\tL_\alpha$ is the restriction of $\Lambda$ to $\A_\alpha$
$(\alpha < \alpha_0) ;$
\item
if we denote by $\bL \in \nchset{\A^{\ast \ast}}{\LHin}$ the normal extension of 
$\Lambda, $
we have $[\bL] = \sup_{\alpha < \alpha_0} [\Lambda_\alpha ] .$
\end{itemize}
Therefore by 
the equivalence \eqref{it:eb2}$\iff$\eqref{it:eb3} of Theorem~\ref{theo:main},
it suffices to show that $\Lambda$ is EB.
For each $\alpha < \alpha_0 ,$
since $\Lambda_\alpha$ is EB, 
so is $\tL_\alpha .$
Hence there exists a symmetric $n$-joint channel
$\Theta_{n ,\alpha} \in \chset{\Atn_\alpha}{\LHin}$
of $\tL_\alpha$
for each $2 \leq n < \infty .$
We define a $\ast$-subalgebra $\tA_n \subset \Atn$
and a map $\Xi_{n, \alpha} \colon \tA_n \to \LHin$
by
\begin{gather*}
	\tA_n := \bigcup_{\alpha < \alpha_0} \Atn_\alpha ,
	\\
	\Xi_{n,\alpha}(X)
	:=
	\begin{cases}
	\Theta_{n,\alpha} (X) 
	& \text{if $X \in \Atn_\alpha;$}
	\\
	0 & \text{otherwise.}
	\end{cases}
\end{gather*}
By Tychonoff's theorem, 
there exists a subnet $(\Xi_{n , \alpha (i)})_{i \in I}$
such that the limit 
$\Theta^\prime_n (X) := \swlim_{i \in I} \Xi_{n, \alpha (i)} (X)$
exists for all $X \in \tA_n .$
Similarly as in the proofs of 
\eqref{it:njoint}$\implies$\eqref{it:ijoints}
and
\eqref{it:ijoints}$\implies$\eqref{it:rand2}
of Theorem~\ref{theo:main},
we can show that $\tA_n$ is norm dense in $\Atn$
and $\Theta_n^\prime$ uniquely extends to a 
symmetric channel $\Theta_n \in \chset{\Atn}{\LHin} .$
Since
\[
	\Theta_n (A \otimes \unit_{\A^{\otimes (n-1)}})
	=\swlim_{i \in I , \alpha(i) \geq \alpha}
	\Theta_{n, \alpha (i)}
	(A \otimes \unit_{\A^{\otimes (n-1)}_{\alpha(i)}})
	=
	\Lambda(A)
\]
for each $\alpha <\alpha_0$ and each
$A \in \A_\alpha ,$
the norm denseness of
$\bigcup_{\alpha < \alpha_0} \A_\alpha$
in $\A$
implies
$\Theta_n (A \otimes \unit_{\A^{\otimes (n-1)}}) = \Lambda (A)$
for all $A\in \A .$
Thus by the symmetry of  $\Theta_n, $
$\Theta_n$ is an $n$-joint channel of $\Lambda .$
Therefore by Theorem~\ref{theo:main}, $\Lambda$ is EB,
which proves the upper Dedekind-closedness of $\CHeb (\LHin) .$

Finally we consider general $\Min .$
We can regard the set $\CH (\Min)$ as the lower subset
\[
	\set{[\Lambda] \in \CH (\LHin) | [\Lambda] \cocp [\id_{\Min}]}
\]
of $\CH (\LHin) .$
Furthermore, by the equivalence \eqref{it:eb1}$\iff$\eqref{it:eb2}
of Theorem~\ref{theo:main},
$\CHeb (\Min)$
can be identified with the lower subset
\[
	\set{[\Lambda] \in \CHeb (\LHin) | [\Lambda] \cocp [\id_{\Min}]}
\]
of $\CHeb (\LHin) \subset \CH (\LHin) .$
Therefore the claim for general $\Min$ follows from that for $\LHin .$
\end{proof}

\section{Example} \label{sec:ex}
In this section, we construct injective normal EB channels 
with arbitrary outcome von Neumann algebras on an infinite-dimensional separable Hilbert space.
We first remark that if we omit the injectivity requirement, 
we have a trivial example of such a channel:
for each von Neumann algebra $\M$ acting on a Hilbert space $\cH $
and a fixed normal state $\vph$ on $\M ,$
the channel
\[
	\M \ni A \mapsto \vph (A) \unit_{\cH} \in  \LH
\]
is an EB channel with the outcome algebra $\M .$
Indeed, the above channel is trivial in the sense
that it is minimal in the randomization order $\cocp $
among $\nchset{}{\LH}.$

In what follows in this section,
we fix an infinite-dimensional separable Hilbert space $\cH $
and an orthonormal basis $(x_n)_{n \in \natz}$ of $\cH ,$
where $\natz$ is the set of natural numbers containing $0 .$
For each $\alpha \in \cmplx$ we define the coherent state vector~\cite{PhysRev.131.2766} by
\[
	\psi_\alpha
	:=
	e^{- |\alpha|^2 /2} \sum_{n \in \natz} \frac{\alpha^n}{\sqrt{n!}} x_n .
\]
The coherent state vectors satisfy the overcompleteness relation~\cite{PhysRev.138.B274}
\[
	\pi^{-1} \int_\cmplx \ket{\psi_\alpha} \bra{\psi_\alpha} d^2 \alpha 
	= \unit_{\cH} ,
\]
where 
$d^2 \alpha = d \Real (\alpha) d \Imag (\alpha) $
is the $2$-dimensional Lebesgue measure and
the integral is in the weak sense.
We denote by $L^p (\cmplx)$ the $L^p$ space of $\cmplx$
with respect to $d^2 \alpha .$
We define normal channels 
$\gbarg \in \nchset{L^\infty(\cmplx)}{\LH} ,$
$\Psi \in \nchset{\LH}{L^\infty(\cmplx)},$
and $\Lambda \colon \nchset{\LH}{\LH}$
by
\begin{gather*}
	\gbarg (f)
	:=
	\pi^{-1} \int_\cmplx f (\alpha) \ket{\psi_\alpha} \bra{\psi_\alpha} d^2 \alpha
	\quad (f \in L^\infty(\cmplx)) , 
	\\
	\Psi (A) (\alpha) := \braket{\psi_\alpha | A \psi_\alpha}
	\quad
	(A \in \LH , \alpha \in \cmplx) ,
	\\
	\Lambda := \gbarg \circ \Psi .
\end{gather*}
The QC channel $\gbarg$ corresponds to the POVM
$\pi^{-1} \ket{\psi_\alpha} \bra{\psi_{\alpha}} d^2 \alpha$
on $\cmplx$ called the Bargmann measure.~\cite{bargmann1961,klauder1968fundamentals}
For each von Neumann algebra $\M$ acting on $\cH ,$
we also define $\Lambda_\M \in \nchset{\M}{\LH} $
as the restriction of $\Lambda$ to $\M .$
Then we immediately have 
$\Lambda_\M \cocp \Lambda \cocp \gbarg $
and hence $\Lambda_\M$ is EB.
Now we have
\begin{proposition}\label{prop:ex}
$\Lambda_\M$ is an injective EB channel for any von Neumann algebra $\M$ acting on $\cH .$
\end{proposition}
\begin{proof}
It suffices to establish the injectivity of $\Lambda .$
Since the injectivity of $\Psi$ is well-known (e.g.\ Ref.~\onlinecite{PhysRev.138.B274}),
we have only to prove the injectivity of $\gbarg .$
Take $f \in L^\infty (\cmplx)$ such that $\gbarg (f) = 0 .$
Then by noting $| \braket{\psi_\alpha | \psi_\beta} |^2 = e^{-|\beta-\alpha|^2} ,$
we have
\[
	\int_\cmplx f(\beta) e^{- | \beta - \alpha |^2} d^2 \beta 
	=
	\pi \braket{ \psi_\alpha| \gbarg (f) \psi_\alpha }
	=
	0
	\quad (\forall \alpha \in \cmplx) 
\]
and hence
\begin{equation}
	\int_\cmplx
	f(\beta)
	(h \ast g) (\beta) d^2 \beta
	=
	\int_\cmplx \int_\cmplx 
	f(\beta)
	e^{-| \beta - \alpha|^2 }
	g(\alpha)
	d^2 \beta d^2 \alpha 
	=0
	\label{eq:fg}
\end{equation}
for any $g \in L^1 (\cmplx) ,$
where $h(\alpha) := e^{- |\alpha |^2}$ and 
\[
	F \ast G (\beta)
	=
	\int_\cmplx F(\beta - \alpha) G(\alpha) d^2 \alpha
\]
is the convolution.
If we take
$g(\alpha) = e^{- | \alpha |^2 + 2ik_1 \Real(\alpha) + 2ik_2 \Imag (\alpha)}$
for each $(k_1, k_2) \in \realn^2 ,$ then a straightforward calculation gives
\[
	h \ast g(\beta)
	=\frac{\pi}{2}
	\exp \left[
	- \frac{| \beta|^2}{2}
	- \frac{1}{2} (k_1^2 + k_2^2)
	+ i k_1 \Real (\beta) + i k_2 \Imag (\beta)
	\right] .
\]
Hence \eqref{eq:fg} implies
\begin{equation}
	\int_\cmplx f(\beta) 
	e^{- \frac{| \beta|^2}{2} + i k_1 \Real (\beta) + i k_2 \Imag (\beta) } 
	d^2 \beta
	= 0
	\label{eq:fourier}
\end{equation}
for any $(k_1 , k_2 ) \in \realn^2 .$
Equation~\eqref{eq:fourier} implies that the Fourier transform 
of the integrable function $f(\beta) e^{- \frac{| \beta|^2}{2} }$ is zero.
Hence, by the injectivity of the Fourier transform, we obtain $f(\alpha) =0$
for almost all $\alpha \in \cmplx, $
proving the injectivity of $\gbarg . $
\end{proof}

By using the concept of the minimal sufficiency of channels,~\cite{kuramochi2017minimal}
we can show that $\Lambda_\M$ is not randomization-equivalent to 
any QC channel if $\M$ is not commutative.
A normal channel $\Gamma \in \nchset{\N}{\Min}$ is called minimal sufficient if
$\Gamma \circ \Phi = \Gamma$ implies $\Phi = \id_{\N}$
for any $\Phi \in \nchset{\N}{\N} .$
Any normal channel $\Gamma$ is randomization-equivalent to a minimal sufficient 
channel and such a minimal sufficient channels is unique up to normal isomorphism 
between outcome algebras.
It is immediate from the definition that any injective normal channel is minimal sufficient
and hence so is $\Lambda_\M .$ 
By following the construction of the minimal sufficient channel equivalent to a given channel (Ref.~\onlinecite{kuramochi2017minimal}, Theorem~1),
we can also show that the minimal sufficient channel equivalent to a QC channel
is also QC.
Therefore $\Lambda_\M$ cannot be randomization-equivalent to a QC channel
if $\M$ is not commutative.

Finally we remark that the injectivity of $\Lambda_\M$ still holds when
$\Psi$ is generalized to 
$\Psi (A) (\alpha)= \braket{\psi_{f(\alpha)} | A \psi_{f(\alpha)}} ,$
where $f : \cmplx \to \cmplx$ is a non-degenerate real affine map.
In specific, if $f (\alpha) = c \overline{\alpha}$ for some positive constant $c ,$
the EB channel $\gbarg \circ \Psi$ corresponds to the conjugate channel of the ideal quantum linear amplification channel
(Ref.~\onlinecite{kuramochi2018directed}, Eq.~(7)).

\section{Concluding remark} \label{sec:conclusion}
In this paper we have investigated the EB condition of channels 
in the infinite-dimensional general operator algebraic framework 
and found equivalent characterizations in Theorem~\ref{theo:main},
generalizing known results in finite dimensions.
Among the conditions, the infinite-self-compatibility condition involves 
the infinite \cstar-tensor product of operator algebras and 
is intrinsically infinite-dimensional.
We have also shown the Dedekind-closedness of normal EB channels
in the context of results of Ref.~\onlinecite{kuramochi2018directed}
and constructed injective normal EB channel with an arbitrary outcome 
von Neumann algebra.

One of the natural generalizations of the present work will be 
the EB condition for positive channels between general(ized) probabilistic theories (GPTs).~\cite{1751-8121-47-32-323001}
In an attempt to such generalization, 
one new problem will be the non-uniqueness of tensor products:~\cite{namioka1969}
in the GPT framework,
the tensor product corresponding to a composite system
is not unique even in finite-dimensions, 
which is not the case for operator algebraic theories.
This kind of problem is related to the definition of EB condition
and the self-compatibility conditions \eqref{it:njoint}-\eqref{it:ijoint}
in Theorem~\ref{theo:main},
while the condition~\eqref{it:rand1} can be straightforwardly generalized.

\begin{acknowledgements}
This work was supported by the National Natural Science Foundation of China 
(Grants No.~11374375 and No.~11574405).
The author would like to thank Erkka Haapasalo for 
helpful comments on the manuscript.
\end{acknowledgements}

\begin{thebibliography}{29}%
\makeatletter
\providecommand \@ifxundefined [1]{%
 \@ifx{#1\undefined}
}%
\providecommand \@ifnum [1]{%
 \ifnum #1\expandafter \@firstoftwo
 \else \expandafter \@secondoftwo
 \fi
}%
\providecommand \@ifx [1]{%
 \ifx #1\expandafter \@firstoftwo
 \else \expandafter \@secondoftwo
 \fi
}%
\providecommand \natexlab [1]{#1}%
\providecommand \enquote  [1]{``#1''}%
\providecommand \bibnamefont  [1]{#1}%
\providecommand \bibfnamefont [1]{#1}%
\providecommand \citenamefont [1]{#1}%
\providecommand \href@noop [0]{\@secondoftwo}%
\providecommand \href [0]{\begingroup \@sanitize@url \@href}%
\providecommand \@href[1]{\@@startlink{#1}\@@href}%
\providecommand \@@href[1]{\endgroup#1\@@endlink}%
\providecommand \@sanitize@url [0]{\catcode `\\12\catcode `\$12\catcode
  `\&12\catcode `\#12\catcode `\^12\catcode `\_12\catcode `\%12\relax}%
\providecommand \@@startlink[1]{}%
\providecommand \@@endlink[0]{}%
\providecommand \url  [0]{\begingroup\@sanitize@url \@url }%
\providecommand \@url [1]{\endgroup\@href {#1}{\urlprefix }}%
\providecommand \urlprefix  [0]{URL }%
\providecommand \Eprint [0]{\href }%
\providecommand \doibase [0]{http://dx.doi.org/}%
\providecommand \selectlanguage [0]{\@gobble}%
\providecommand \bibinfo  [0]{\@secondoftwo}%
\providecommand \bibfield  [0]{\@secondoftwo}%
\providecommand \translation [1]{[#1]}%
\providecommand \BibitemOpen [0]{}%
\providecommand \bibitemStop [0]{}%
\providecommand \bibitemNoStop [0]{.\EOS\space}%
\providecommand \EOS [0]{\spacefactor3000\relax}%
\providecommand \BibitemShut  [1]{\csname bibitem#1\endcsname}%
\let\auto@bib@innerbib\@empty
\bibitem [{\citenamefont {Horodecki}, \citenamefont {Shor},\ and\ \citenamefont
  {Ruskai}(2003)}]{Horodecki2003}%
  \BibitemOpen
  \bibfield  {author} {\bibinfo {author} {\bibfnamefont {M.}~\bibnamefont
  {Horodecki}}, \bibinfo {author} {\bibfnamefont {P.~W.}\ \bibnamefont {Shor}},
  \ and\ \bibinfo {author} {\bibfnamefont {M.~B.}\ \bibnamefont {Ruskai}},\
  }\bibfield  {title} {\enquote {\bibinfo {title} {Entanglement breaking
  channels},}\ }\href {\doibase 10.1142/S0129055X03001709} {\bibfield
  {journal} {\bibinfo  {journal} {Rev. Math. Phys.}\ }\textbf {\bibinfo
  {volume} {15}},\ \bibinfo {pages} {629--641} (\bibinfo {year}
  {2003})}\BibitemShut {NoStop}%
\bibitem [{\citenamefont {Holevo}, \citenamefont {Shirokov},\ and\
  \citenamefont {Werner}(2005)}]{holevo2005separability}%
  \BibitemOpen
  \bibfield  {author} {\bibinfo {author} {\bibfnamefont {A.~S.}\ \bibnamefont
  {Holevo}}, \bibinfo {author} {\bibfnamefont {M.~E.}\ \bibnamefont
  {Shirokov}}, \ and\ \bibinfo {author} {\bibfnamefont {R.~F.}\ \bibnamefont
  {Werner}},\ }\bibfield  {title} {\enquote {\bibinfo {title} {Separability and
  entanglement-breaking in infinite dimensions},}\ }\href@noop {} {\bibfield
  {journal} {\bibinfo  {journal} {arXiv preprint quant-ph/0504204}\ } (\bibinfo
  {year} {2005})}\BibitemShut {NoStop}%
\bibitem [{\citenamefont {Kholevo}, \citenamefont {Shirokov},\ and\
  \citenamefont {Werner}(2005)}]{0036-0279-60-2-L12}%
  \BibitemOpen
  \bibfield  {author} {\bibinfo {author} {\bibfnamefont {A.~S.}\ \bibnamefont
  {Kholevo}}, \bibinfo {author} {\bibfnamefont {M.~E.}\ \bibnamefont
  {Shirokov}}, \ and\ \bibinfo {author} {\bibfnamefont {R.~F.}\ \bibnamefont
  {Werner}},\ }\bibfield  {title} {\enquote {\bibinfo {title} {{On the notion
  of entanglement in Hilbert spaces}},}\ }\href
  {http://stacks.iop.org/0036-0279/60/i=2/a=L12} {\bibfield  {journal}
  {\bibinfo  {journal} {Russian Math. Surveys}\ }\textbf {\bibinfo {volume}
  {60}},\ \bibinfo {pages} {359} (\bibinfo {year} {2005})}\BibitemShut
  {NoStop}%
\bibitem [{\citenamefont {Holevo}(2008)}]{Holevo2008}%
  \BibitemOpen
  \bibfield  {author} {\bibinfo {author} {\bibfnamefont {A.~S.}\ \bibnamefont
  {Holevo}},\ }\bibfield  {title} {\enquote {\bibinfo {title}
  {Entanglement-breaking channels in infinite dimensions},}\ }\href {\doibase
  10.1134/S0032946008030010} {\bibfield  {journal} {\bibinfo  {journal} {Probl.
  Inf. Transm.}\ }\textbf {\bibinfo {volume} {44}},\ \bibinfo {pages}
  {171--184} (\bibinfo {year} {2008})}\BibitemShut {NoStop}%
\bibitem [{\citenamefont {Holevo}(2011{\natexlab{a}})}]{Holevo2011}%
  \BibitemOpen
  \bibfield  {author} {\bibinfo {author} {\bibfnamefont {A.~S.}\ \bibnamefont
  {Holevo}},\ }\bibfield  {title} {\enquote {\bibinfo {title} {{Entropy gain
  and the Choi-Jamiolkowski correspondence for infinite-dimensional quantum
  evolutions}},}\ }\href {\doibase 10.1007/s11232-011-0010-5} {\bibfield
  {journal} {\bibinfo  {journal} {Theor. Math. Phys.}\ }\textbf {\bibinfo
  {volume} {166}},\ \bibinfo {pages} {123--138} (\bibinfo {year}
  {2011}{\natexlab{a}})}\BibitemShut {NoStop}%
\bibitem [{\citenamefont {He}(2013)}]{He2013}%
  \BibitemOpen
  \bibfield  {author} {\bibinfo {author} {\bibfnamefont {K.}~\bibnamefont
  {He}},\ }\bibfield  {title} {\enquote {\bibinfo {title} {On entanglement
  breaking channels for infinite dimensional quantum systems},}\ }\href
  {\doibase 10.1007/s10773-012-1303-7} {\bibfield  {journal} {\bibinfo
  {journal} {Int. J. Theor. Phys.}\ }\textbf {\bibinfo {volume} {52}},\
  \bibinfo {pages} {1886--1892} (\bibinfo {year} {2013})}\BibitemShut {NoStop}%
\bibitem [{\citenamefont {St{\o}rmer}(2008)}]{STORMER20082303}%
  \BibitemOpen
  \bibfield  {author} {\bibinfo {author} {\bibfnamefont {E.}~\bibnamefont
  {St{\o}rmer}},\ }\bibfield  {title} {\enquote {\bibinfo {title} {Separable
  states and positive maps},}\ }\href {\doibase
  https://doi.org/10.1016/j.jfa.2007.12.017} {\bibfield  {journal} {\bibinfo
  {journal} {J. Func. Anal.}\ }\textbf {\bibinfo {volume} {254}},\ \bibinfo
  {pages} {2303 -- 2312} (\bibinfo {year} {2008})}\BibitemShut {NoStop}%
\bibitem [{\citenamefont {St{\o}rmer}(2009)}]{Stoermer_2009}%
  \BibitemOpen
  \bibfield  {author} {\bibinfo {author} {\bibfnamefont {E.}~\bibnamefont
  {St{\o}rmer}},\ }\bibfield  {title} {\enquote {\bibinfo {title} {Separable
  states and positive maps {II}},}\ }\href {\doibase
  10.7146/math.scand.a-15114} {\bibfield  {journal} {\bibinfo  {journal} {Math.
  SCAND.}\ }\textbf {\bibinfo {volume} {105}},\ \bibinfo {pages} {188--198}
  (\bibinfo {year} {2009})}\BibitemShut {NoStop}%
\bibitem [{\citenamefont {Rahaman}, \citenamefont {Jaques},\ and\ \citenamefont
  {Paulsen}(2018)}]{doi:10.1063/1.5024385}%
  \BibitemOpen
  \bibfield  {author} {\bibinfo {author} {\bibfnamefont {M.}~\bibnamefont
  {Rahaman}}, \bibinfo {author} {\bibfnamefont {S.}~\bibnamefont {Jaques}}, \
  and\ \bibinfo {author} {\bibfnamefont {V.~I.}\ \bibnamefont {Paulsen}},\
  }\bibfield  {title} {\enquote {\bibinfo {title} {Eventually entanglement
  breaking maps},}\ }\href {\doibase 10.1063/1.5024385} {\bibfield  {journal}
  {\bibinfo  {journal} {J. Math. Phys.}\ }\textbf {\bibinfo {volume} {59}},\
  \bibinfo {pages} {062201} (\bibinfo {year} {2018})}\BibitemShut {NoStop}%
\bibitem [{\citenamefont {Heinosaari}\ and\ \citenamefont
  {Miyadera}(2017)}]{1751-8121-50-13-135302}%
  \BibitemOpen
  \bibfield  {author} {\bibinfo {author} {\bibfnamefont {T.}~\bibnamefont
  {Heinosaari}}\ and\ \bibinfo {author} {\bibfnamefont {T.}~\bibnamefont
  {Miyadera}},\ }\bibfield  {title} {\enquote {\bibinfo {title}
  {Incompatibility of quantum channels},}\ }\href
  {http://stacks.iop.org/1751-8121/50/i=13/a=135302} {\bibfield  {journal}
  {\bibinfo  {journal} {J. Phys. A: Math. Theor.}\ }\textbf {\bibinfo {volume}
  {50}},\ \bibinfo {pages} {135302} (\bibinfo {year} {2017})}\BibitemShut
  {NoStop}%
\bibitem [{\citenamefont {Chiribella}(2011)}]{10.1007/978-3-642-18073-6_2}%
  \BibitemOpen
  \bibfield  {author} {\bibinfo {author} {\bibfnamefont {G.}~\bibnamefont
  {Chiribella}},\ }\bibfield  {title} {\enquote {\bibinfo {title} {{On Quantum
  Estimation, Quantum Cloning and Finite Quantum de Finetti Theorems}},}\ }in\
  \href@noop {} {\emph {\bibinfo {booktitle} {Theory of Quantum Computation,
  Communication, and Cryptography}}},\ \bibinfo {editor} {edited by\ \bibinfo
  {editor} {\bibfnamefont {W.}~\bibnamefont {van Dam}}, \bibinfo {editor}
  {\bibfnamefont {V.~M.}\ \bibnamefont {Kendon}}, \ and\ \bibinfo {editor}
  {\bibfnamefont {S.}~\bibnamefont {Severini}}}\ (\bibinfo  {publisher}
  {Springer},\ \bibinfo {address} {Berlin, Heidelberg},\ \bibinfo {year}
  {2011})\ pp.\ \bibinfo {pages} {9--25}\BibitemShut {NoStop}%
\bibitem [{\citenamefont {Takesaki}(1979)}]{takesakivol1}%
  \BibitemOpen
  \bibfield  {author} {\bibinfo {author} {\bibfnamefont {M.}~\bibnamefont
  {Takesaki}},\ }\href@noop {} {\emph {\bibinfo {title} {{Theory of Operator
  Algebras I}}}}\ (\bibinfo  {publisher} {Springer},\ \bibinfo {year}
  {1979})\BibitemShut {NoStop}%
\bibitem [{\citenamefont {Sakai}(1971)}]{sakaibook}%
  \BibitemOpen
  \bibfield  {author} {\bibinfo {author} {\bibfnamefont {S.}~\bibnamefont
  {Sakai}},\ }\href@noop {} {\emph {\bibinfo {title} {{$C^\ast$-Algebras and
  $W^\ast$-Algebras}}}}\ (\bibinfo  {publisher} {Springer},\ \bibinfo {year}
  {1971})\BibitemShut {NoStop}%
\bibitem [{\citenamefont
  {Kuramochi}(2018{\natexlab{a}})}]{kuramochi2018incomp}%
  \BibitemOpen
  \bibfield  {author} {\bibinfo {author} {\bibfnamefont {Y.}~\bibnamefont
  {Kuramochi}},\ }\bibfield  {title} {\enquote {\bibinfo {title} {Quantum
  incompatibility of channels with general outcome operator algebras},}\ }\href
  {\doibase 10.1063/1.5008300} {\bibfield  {journal} {\bibinfo  {journal} {J.
  Math. Phys.}\ }\textbf {\bibinfo {volume} {59}},\ \bibinfo {pages} {042203}
  (\bibinfo {year} {2018}{\natexlab{a}})}\BibitemShut {NoStop}%
\bibitem [{\citenamefont {Davies}\ and\ \citenamefont
  {Lewis}(1970)}]{davieslewisBF01647093}%
  \BibitemOpen
  \bibfield  {author} {\bibinfo {author} {\bibfnamefont {E.~B.}\ \bibnamefont
  {Davies}}\ and\ \bibinfo {author} {\bibfnamefont {J.}~\bibnamefont {Lewis}},\
  }\bibfield  {title} {\enquote {\bibinfo {title} {An operational approach to
  quantum probability},}\ }\href {\doibase 10.1007/BF01647093} {\bibfield
  {journal} {\bibinfo  {journal} {Commun. Math. Phys.}\ }\textbf {\bibinfo
  {volume} {17}},\ \bibinfo {pages} {239--260} (\bibinfo {year}
  {1970})}\BibitemShut {NoStop}%
\bibitem [{\citenamefont
  {Holevo}(2011{\natexlab{b}})}]{holevo2011probabilistic}%
  \BibitemOpen
  \bibfield  {author} {\bibinfo {author} {\bibfnamefont {A.~S.}\ \bibnamefont
  {Holevo}},\ }\href@noop {} {\emph {\bibinfo {title} {Probabilistic and
  statistical aspects of quantum theory}}}\ (\bibinfo  {publisher} {Springer},\
  \bibinfo {year} {2011})\BibitemShut {NoStop}%
\bibitem [{\citenamefont {Busch}\ \emph {et~al.}(2016)\citenamefont {Busch},
  \citenamefont {Lahti}, \citenamefont {Pellonp{\"a}{\"a}},\ and\ \citenamefont
  {Ylinen}}]{busch2016quantum}%
  \BibitemOpen
  \bibfield  {author} {\bibinfo {author} {\bibfnamefont {P.}~\bibnamefont
  {Busch}}, \bibinfo {author} {\bibfnamefont {P.~J.}\ \bibnamefont {Lahti}},
  \bibinfo {author} {\bibfnamefont {J.-P.}\ \bibnamefont {Pellonp{\"a}{\"a}}},
  \ and\ \bibinfo {author} {\bibfnamefont {K.}~\bibnamefont {Ylinen}},\
  }\href@noop {} {\emph {\bibinfo {title} {Quantum measurement}}}\ (\bibinfo
  {publisher} {Springer},\ \bibinfo {year} {2016})\BibitemShut {NoStop}%
\bibitem [{\citenamefont {Takeda}(1955)}]{takeda1955}%
  \BibitemOpen
  \bibfield  {author} {\bibinfo {author} {\bibfnamefont {Z.}~\bibnamefont
  {Takeda}},\ }\bibfield  {title} {\enquote {\bibinfo {title} {Inductive limit
  and infinite direct product of operator algebras},}\ }\href {\doibase
  10.2748/tmj/1178245105} {\bibfield  {journal} {\bibinfo  {journal} {Tohoku
  Math. J. (2)}\ }\textbf {\bibinfo {volume} {7}},\ \bibinfo {pages} {67--86}
  (\bibinfo {year} {1955})}\BibitemShut {NoStop}%
\bibitem [{\citenamefont
  {Kuramochi}(2018{\natexlab{b}})}]{kuramochi2018access}%
  \BibitemOpen
  \bibfield  {author} {\bibinfo {author} {\bibfnamefont {Y.}~\bibnamefont
  {Kuramochi}},\ }\bibfield  {title} {\enquote {\bibinfo {title} {Accessible
  information without disturbing partially known quantum states on a von
  neumann algebra},}\ }\href {\doibase 10.1007/s10773-018-3749-8} {\bibfield
  {journal} {\bibinfo  {journal} {Int. J. Theor. Phys.}\ }\textbf {\bibinfo
  {volume} {57}},\ \bibinfo {pages} {2249--2266} (\bibinfo {year}
  {2018}{\natexlab{b}})}\BibitemShut {NoStop}%
\bibitem [{\citenamefont {Phelps}(2001)}]{phelps2001lectures}%
  \BibitemOpen
  \bibfield  {author} {\bibinfo {author} {\bibfnamefont {R.~R.}\ \bibnamefont
  {Phelps}},\ }\href@noop {} {\emph {\bibinfo {title} {{Lectures on Choquet's
  theorem}}}}\ (\bibinfo  {publisher} {Springer},\ \bibinfo {year}
  {2001})\BibitemShut {NoStop}%
\bibitem [{\citenamefont {Raginsky}(2003)}]{doi:10.1063/1.1615697}%
  \BibitemOpen
  \bibfield  {author} {\bibinfo {author} {\bibfnamefont {M.}~\bibnamefont
  {Raginsky}},\ }\bibfield  {title} {\enquote {\bibinfo {title} {{Radon-Nikodym
  derivatives of quantum operations}},}\ }\href {\doibase 10.1063/1.1615697}
  {\bibfield  {journal} {\bibinfo  {journal} {J. Math. Phys.}\ }\textbf
  {\bibinfo {volume} {44}},\ \bibinfo {pages} {5003--5020} (\bibinfo {year}
  {2003})}\BibitemShut {NoStop}%
\bibitem [{\citenamefont
  {Kuramochi}(2018{\natexlab{c}})}]{kuramochi2018directed}%
  \BibitemOpen
  \bibfield  {author} {\bibinfo {author} {\bibfnamefont {Y.}~\bibnamefont
  {Kuramochi}},\ }\bibfield  {title} {\enquote {\bibinfo {title}
  {Directed-completeness of quantum statistical experiments in the
  randomization order},}\ }\href@noop {} {\bibfield  {journal} {\bibinfo
  {journal} {arXiv preprint arXiv:1805.04357}\ } (\bibinfo {year}
  {2018}{\natexlab{c}})}\BibitemShut {NoStop}%
\bibitem [{\citenamefont {Glauber}(1963)}]{PhysRev.131.2766}%
  \BibitemOpen
  \bibfield  {author} {\bibinfo {author} {\bibfnamefont {R.~J.}\ \bibnamefont
  {Glauber}},\ }\bibfield  {title} {\enquote {\bibinfo {title} {Coherent and
  incoherent states of the radiation field},}\ }\href {\doibase
  10.1103/PhysRev.131.2766} {\bibfield  {journal} {\bibinfo  {journal} {Phys.
  Rev.}\ }\textbf {\bibinfo {volume} {131}},\ \bibinfo {pages} {2766--2788}
  (\bibinfo {year} {1963})}\BibitemShut {NoStop}%
\bibitem [{\citenamefont {Mehta}\ and\ \citenamefont
  {Sudarshan}(1965)}]{PhysRev.138.B274}%
  \BibitemOpen
  \bibfield  {author} {\bibinfo {author} {\bibfnamefont {C.~L.}\ \bibnamefont
  {Mehta}}\ and\ \bibinfo {author} {\bibfnamefont {E.~C.~G.}\ \bibnamefont
  {Sudarshan}},\ }\bibfield  {title} {\enquote {\bibinfo {title} {Relation
  between quantum and semiclassical description of optical coherence},}\ }\href
  {\doibase 10.1103/PhysRev.138.B274} {\bibfield  {journal} {\bibinfo
  {journal} {Phys. Rev.}\ }\textbf {\bibinfo {volume} {138}},\ \bibinfo {pages}
  {B274--B280} (\bibinfo {year} {1965})}\BibitemShut {NoStop}%
\bibitem [{\citenamefont {Bargmann}(1961)}]{bargmann1961}%
  \BibitemOpen
  \bibfield  {author} {\bibinfo {author} {\bibfnamefont {V.}~\bibnamefont
  {Bargmann}},\ }\bibfield  {title} {\enquote {\bibinfo {title} {{On a Hilbert
  space of analytic functions and an associated integral transform part I}},}\
  }\href {\doibase 10.1002/cpa.3160140303} {\bibfield  {journal} {\bibinfo
  {journal} {Comm. Pure Appl. Math.}\ }\textbf {\bibinfo {volume} {14}},\
  \bibinfo {pages} {187--214} (\bibinfo {year} {1961})}\BibitemShut {NoStop}%
\bibitem [{\citenamefont {Klauder}\ and\ \citenamefont
  {Sudarshan}(1968)}]{klauder1968fundamentals}%
  \BibitemOpen
  \bibfield  {author} {\bibinfo {author} {\bibfnamefont {J.}~\bibnamefont
  {Klauder}}\ and\ \bibinfo {author} {\bibfnamefont {G.}~\bibnamefont
  {Sudarshan}},\ }\href@noop {} {\emph {\bibinfo {title} {{Fundamentals of
  Quantum Optics}}}}\ (\bibinfo  {publisher} {Benjamin, New York},\ \bibinfo
  {year} {1968})\BibitemShut {NoStop}%
\bibitem [{\citenamefont {Kuramochi}(2017)}]{kuramochi2017minimal}%
  \BibitemOpen
  \bibfield  {author} {\bibinfo {author} {\bibfnamefont {Y.}~\bibnamefont
  {Kuramochi}},\ }\bibfield  {title} {\enquote {\bibinfo {title} {{Minimal
  sufficient statistical experiments on von Neumann algebras}},}\ }\href
  {\doibase 10.1063/1.4986247} {\bibfield  {journal} {\bibinfo  {journal} {J.
  Math. Phys.}\ }\textbf {\bibinfo {volume} {58}},\ \bibinfo {pages} {062203}
  (\bibinfo {year} {2017})}\BibitemShut {NoStop}%
\bibitem [{\citenamefont {Janotta}\ and\ \citenamefont
  {Hinrichsen}(2014)}]{1751-8121-47-32-323001}%
  \BibitemOpen
  \bibfield  {author} {\bibinfo {author} {\bibfnamefont {P.}~\bibnamefont
  {Janotta}}\ and\ \bibinfo {author} {\bibfnamefont {H.}~\bibnamefont
  {Hinrichsen}},\ }\bibfield  {title} {\enquote {\bibinfo {title} {Generalized
  probability theories: what determines the structure of quantum theory?}}\
  }\href {http://stacks.iop.org/1751-8121/47/i=32/a=323001} {\bibfield
  {journal} {\bibinfo  {journal} {J. Phys. A: Math. Theor.}\ }\textbf {\bibinfo
  {volume} {47}},\ \bibinfo {pages} {323001} (\bibinfo {year}
  {2014})}\BibitemShut {NoStop}%
\bibitem [{\citenamefont {Namioka}\ and\ \citenamefont
  {Phelps}(1969)}]{namioka1969}%
  \BibitemOpen
  \bibfield  {author} {\bibinfo {author} {\bibfnamefont {I.}~\bibnamefont
  {Namioka}}\ and\ \bibinfo {author} {\bibfnamefont {R.~R.}\ \bibnamefont
  {Phelps}},\ }\bibfield  {title} {\enquote {\bibinfo {title} {Tensor products
  of compact convex sets.}}\ }\href
  {https://projecteuclid.org:443/euclid.pjm/1102977882} {\bibfield  {journal}
  {\bibinfo  {journal} {Pacific J. Math.}\ }\textbf {\bibinfo {volume} {31}},\
  \bibinfo {pages} {469--480} (\bibinfo {year} {1969})}\BibitemShut {NoStop}%
\end{thebibliography}
%

\end{document}